\documentclass[10pt]{article}

\usepackage{latexsym}
\usepackage{verbatim}
\usepackage{amsmath,amssymb,amsthm}
\usepackage{enumerate}
\usepackage[dvips,draft,final]{graphics}

\newcommand{\FF}{\mathcal{F}}
\def\C{\mathbb{C}}
\def\Z{\mathbb{Z}}
\def\N{\mathbb{N}}
\def\R{\mathbb{R}}

\newcommand{\PP}{\mathcal P}
\def\CC{{\mathcal C}}

\def\e{{\rm e}}

\def\eps{\varepsilon}

\def\ni{\noindent}
\def\op{{\rm op}}

\def\supp{{\rm supp\,}}
\def\ol{\overline}
\def\d{{\rm d}}

\def\op_#1{\mathrel{\mathop{{\rm op}_{#1}}}}

\def\build#1_#2^#3{\mathrel{
\mathop{\kern 0pt#1}\limits_{#2}^{#3}}}

\def\td_#1,#2{\mathrel{
\mathop{\build\longrightarrow_{#1\rightarrow #2}^{}}}}

\def\lim_#1,#2{\mathrel{
\mathop{\build{\rm lim}_{#1\rightarrow#2}^{}}}}

\def\limsup_#1,#2{\mathrel{
\mathop{\build{\rm limsup}_{#1\rightarrow#2}^{}}}}

\def\liminf_#1,#2{\mathrel{
\mathop{\build{\rm liminf}_{#1\rightarrow#2}^{}}}}

\def\aref#1{(\ref{#1})}
\def\eps{\varepsilon}

\newtheorem{prop}{Proposition}[section]
\newtheorem{theorem}[prop]{Theorem}
\newtheorem{lem}[prop]{Lemma}
\theoremstyle{definition}
\newtheorem{defi}[prop]{Definition}
\newtheorem{rem}[prop]{Remark}

\def\1{{\bf 1}}
\def\0{{\bf 0}}

\def\ni{\noindent}
\def\op{{\rm op}}
\def\ol{\overline}
\def\d{{\rm d}}

\numberwithin{equation}{section} 

\begin{document}

\title{Degenerated codimension~1 crossings and resolvent estimates}

\author{Thomas \textsc{Duyckaerts}
\footnote{Universit\'e de Cergy-Pontoise, UMR CNRS 8088. Partially supported by the French ANR grants OndNonLin and ControlFlux},
\\
Clotilde \textsc{Fermanian Kammerer}
\footnote{Universit\'e Paris Est, UMR CNRS 8050},
\\
Thierry \textsc{Jecko}\footnote{Universit\'e de Cergy-Pontoise, UMR CNRS 8088. Supported by the 
French ANR grant BLAN08-3\_ 309070}}

\date{}

\maketitle

{\small{\bf Abstract}: {\it In this article, we analyze the propagation of Wigner measures of a family of solutions to a system of  semi-classical pseudodifferential equations presenting eigenvalues crossings on hypersurfaces.   We prove the propagation along classical trajectories under a geometric condition which is satisfied for example as soon as the Hamiltonian vector fields  are transverse or tangent at finite order to the crossing set. We derive resolvent estimates for semi-classical Schr\"odinger operator with matrix-valued potential under a geometric condition of the same type on the crossing set and we analyze examples of degenerate situations where one can prove transfers between the modes.
}}

\section{Introduction and main results}

\subsection{Resolvent estimate and nontrapping condition}

In this paper, we are concerned with the semi-classical
Schr\"odinger operator acting in $L^2(\R^d,\C^N)$,
\begin{equation}
P(\eps) =-\frac{\eps^2}{2}\Delta \cdot {\rm Id}+M(x)
\end{equation}
where ${\rm Id}$ is the $N\times N$ identity matrix and $M$ is a smooth, self-adjoint $N\times N$ matrix-valued potential. We require a long range behaviour of $M$: the matrix~$M$ has a
limit $M_\infty$ at infinity and there exists $\rho>0$ such that
\begin{eqnarray}
 \label{hyp:lr}
  &  \forall
  \alpha\in\N^d,\;\, \forall x\in\R^d,\;\;\left|\partial_x^\alpha\, (M(x)-M_\infty)\right|=O_\alpha\left(\langle x\rangle ^{-\rho-|\alpha|}\right),&
\end{eqnarray}

\ni where $\langle x\rangle =(1+|x|^2)^{1/2}$. 
The operator $P(\eps)$ is self-adjoint with domain $H^2(\R^d,\C^N)$ (see \cite{RS2}). We define its resolvent
$$R(z,\eps)=(P(\eps)-z)^{-1}$$ for $z$ in the resolvent set.
As already seen in \cite{J6}, the limiting absorption principle is valid on any interval
$I\Subset (\|M_\infty\|_\infty;+\infty)$. This means that, for any $s>1/2$ and any fixed $\eps >0$,
\begin{equation}\label{lap}
\build\sup _{\textrm{Re}\, z\in I\, ,\, \textrm{Im}\, z\neq 0}^{}\|\langle x\rangle^{-s}
R(z,\eps )\langle x\rangle^{-s}\|_{L^2\rightarrow L^2}\ <\ \infty \, .
\end{equation}
It turns out that the size of the previous supremum with respect to $\eps$ is
important for scattering theory. Roughly speaking, it is $O(\eps
^{-1})$ for non resonant scattering and larger, up to some $\exp (c\eps
^{-1})$ with $c>0$, when a resonance is present. In the scalar case (cf. \cite{rt,w87,w}),
the non resonant situation is characterized in term of a nontrapping condition on the classical
trajectories of the Hamiltonian field associated to the operator. In the matricial case, the 
treatment is much more complicated. Known results depend on the codimension of the eigenvalues 
crossings: see \cite{J1} (no crossing), \cite{J2,J6} (codimension one crossing), and \cite{FR}
(codimension two crossing). Here we focus on the second situation.

%

\ni For codimension one eigenvalues crossing, an assumption on the
matricial structure of $M$ was used in \cite{J6}. Our aim in
the present paper is to weaken this assumption as much as possible.
As in \cite{J6}, we focus on the situation where there exists $m\in \{1,\ldots, N\}$ such that 
\begin{equation}
\label{formeM}
M(x)=\sum_{j=1}^m E_j(x)\Pi_j(x), 
\end{equation}
where, denoting by $\C^{N,N}$ the algebra of $N\times N$ complex matrices, 
\begin{equation}
\label{hyp:regularity}
E_1, \ldots ,E_m \in C^{\infty}\left(\R^d,\R\right), \quad \Pi_1,\ldots,\Pi_m\in C^{\infty}\left(\R^d,\C^{N,N}\right),
\end{equation}
\begin{equation}
\label{hyp:projectors}
\text{and for all $x$, the $\Pi_j(x)$ are orthogonal projectors satisfying $\Pi_j(x)\Pi_k(x)=0$ if $j\neq k$.}
\end{equation}

\ni We will also assume that if $j\neq k$, $E_j$ differs from $E_k$ at least in an large open subset of $\R^d$.
 In~\cite{J6} it is shown that this situation occurs under reasonable assumptions in the case of codimension $1$ crossing.
 We will call it a \textit{smooth crossing} by opposition to situations (occurring in codimension $2$ or $3$ crossings) where
 the eigenvalues and/or the eigenprojections develop singularities.

\ni The \textit{crossing set} is the closed subset of $\R^d$
$$\CC=\{x\in \R^d;\; \exists j,k; \,\; E_j(x)=E_k(x)\text{ and } j\neq k\}.$$
 We will also call \textit{crossing set} the corresponding subset $\Gamma$ of $T^*\R^d$:
$$ \Gamma=\{(x,\xi)\in T^{*}\R^d, \; x\in \CC\}=\bigcup_{j=1}^m \Gamma_j,\text{ where } \Gamma_j=\{(x,\xi)\in T^*\R^d;\;
 \exists k\neq j,\; E_j(x)=E_k(x)\}.$$
\ni We say that the crossing is of codimension~$1$ when~$\CC$ is a smooth codimension~$1$ submanifold of~$\R^d$,
which is equivalent to say that~$\Gamma$ is a smooth codimension~$1$ submanifold of~$T^*\R^d$.
In the sequel we will make a weaker assumption, assuming that each of the~$\Gamma_j$ is included in a codimension~$1$
submanifold~$\Sigma_j$. This covers of course the case of codimension~$1$ crossing, but also the cases of
codimension~$2$ and~$3$ crossing, when the crossing is smooth. We emphasize however that this smoothness condition is
not satisfied for generic codimension~$2$ and~$3$ crossing, and that our result has more impact in the codimension~$1$ case.

%
\ni Consider for $j=1\ldots m$ the eigenvalues $\lambda_j(x,\xi)=\frac 12|\xi|^2+E_j(x)$ of the semi-classical symbol of $P(\eps)$. We denote its  Hamilton field by
$$H_{j}(x,\xi )=\left(\nabla _\xi\lambda_j (x,\xi ),-\nabla _x\lambda_j (x,\xi )
\right)=\left(\xi,-\nabla_xE_j(x)\right).$$

\ni Let~$\Sigma _j(\infty ;H_j)$ be the set of points~$(x,\xi)$ in~$\Sigma_j$ where~$H_j$ is tangent
at infinite order to~$\Sigma_j$. If~$\Sigma_j$ is given (locally) by the equation~$\gamma_j(x)=0$,
then~$\Sigma _j(\infty ;H_j)$ is (locally) the set of point such that~$H_j^{k}\gamma_j(x)=0$, for all~$k$ in~$\N$.

\ni Our first result gives the desired characterization of the non-resonant situation under a purely geometric condition.
For the sake of clarity we do not give yet the stronger condition possible (see Theorem~\ref{theo:resest} below).

\ni Let $\rho_j (t;x^\ast ,\xi^\ast )=\left(x_j(t;x^\ast,\xi^\ast),\xi_j(t;x^\ast,\xi^\ast)\right)$ be the maximal solution of the Hamilton equation $\dot{\rho}_j:=d\rho_j/dt=H_j(\rho_j)$
with initial condition $\rho_j(0)=(x^\ast ;\xi ^\ast )$. Observe that the flow $\rho _j$ is complete since $E_j$
is smooth. We say that  $\lambda_j$ (or $H_j$ or $\rho_j$) is
\textit{non-trapping} at energy $E\in\R$ if, for all $(x^\ast ;\xi ^\ast )\in
\lambda_j ^{-1}(E)$,
\begin{equation}\label{non-capture}
\lim_{|t|},{\infty}\left|x_j(t;x^\ast ;\xi^\ast )\right|=+\infty.
\end{equation}

\begin{theorem}\label{theo:resestbis}
Let $M$ satisfy \eqref{hyp:lr}, \eqref{formeM}, \eqref{hyp:regularity}, and \eqref{hyp:projectors}. Assume that for all $j$,
$\Gamma _j$ is included in a smooth submanifold $\Sigma_j$ of codimension $1$ of $T^*\R^d$. Take an open interval
$I_0\subset(\|M_\infty\|_\infty;+\infty)$. Assume that, for all $j\in \{1,\cdots ,m\}$, $\Sigma _j(\infty ;H_j)\cap\lambda_j^{-1}(I_0)$
is finite or countable. Then, the following conditions are equivalent
\begin{enumerate}
\item\label{nontrapping} for all $j\in \{1,\cdots ,m\}$
and all $E\in I_0$, $\lambda _j$ is non trapping at energy $E$;
\item\label{resolvent_estimate} for all interval $I\Subset I_0$, for all $s>1/2$, there exist $\eps_0>0$ and $C_{s,I}>0$
such that, for all $\epsilon\in(0;\epsilon _0]$,
\begin{equation}\label{resest2}
\build\sup _{\textrm{Re}\, z\in I\, ,\, \textrm{Im}\, z\neq 0}^{}\|\langle x\rangle^{-s}
R(z,\eps )\langle x\rangle^{-s}\|_{L^2\to L^2}\ \leq \ \frac{C_{s,I}}{\eps}\, .
\end{equation}
\end{enumerate}
\end{theorem}
\begin{rem}
In particular, if for all $j$ the flow $H_{j}$ is (outside
$\{\xi=0\}$) transverse or tangent to finite order to the $j$th
crossing set $\Gamma _j$, then the conditions (\ref{nontrapping})
and (\ref{resolvent_estimate}) are equivalent.
\end{rem}

\ni Actually, we can construct, for each $j$, some subset $\FF_{j}$ of $\Sigma _j(\infty ;H_j)$, depending on $\Sigma _j$ and on $H_{j}$, such that Theorem~\ref{theo:resestbis} still holds true when the countability condition on $\Sigma _j(\infty ;H_j)\cap\lambda_j^{-1}(I_0)$ is replaced by the less restrictive condition $\FF_j\cap \lambda_j^{-1}(I_0)=\emptyset$.
A precise definition of the set $\FF_j$ is given
in Subsection~\ref{proof17} via Definitions~\ref{def:F(alpha)} and \ref{def:FF}.

\ni The conditions of Theorem~\ref{theo:resestbis} (and also of its refined version with the condition on $\FF_{j}$) bear only on the eigenvalues $\lambda_j$ and do not depend on the projectors~$\Pi_j$. We will now combine them with the condition on the matricial structure of~$M$ which was introduced in~\cite{J6}.
The potential~$M$ satisfies the~\textit{special condition} at the
crossing if, for all~$j\in\{1,\cdots ,m\}$, the projector~$\Pi_j$ is conormal
to~$\CC$, that is, for any~$x\in\CC$ and~$\xi$ in the
tangent space~$T_{x}\CC$ of~$\CC$ at~$x$, $\xi\cdot \nabla \Pi_j (x)=0$. Under this special condition at the
crossing and some technical one at infinity (in the~$x$ variable), that may be removed by the arguments of~\cite{FR}, it was proved in~\cite{J6},
that the conditions~\eqref{nontrapping} and~\eqref{resolvent_estimate} of Theorem~\ref{theo:resestbis} are equivalent.
The following theorem implies Theorem~\ref{theo:resestbis} and the result of~\cite{J6}.


%
\begin{theorem}\label{theo:resest}
Let $M$ satisfy \eqref{hyp:lr}, \eqref{formeM},
\eqref{hyp:regularity} and \eqref{hyp:projectors}. Let~$I_0$ be an
open interval included in~$(\|M_\infty\|_\infty;+\infty)$ and
assume that, for all $j\in \{1,\cdots ,m\}$, the function $(x,\xi
)\mapsto \xi\cdot \nabla \Pi _j(x)$ vanishes on
$\FF_j\cap\lambda_j^{-1}(I_0)$. Then, the conditions~\eqref{nontrapping} and~\eqref{resolvent_estimate} of
Theorem~\ref{theo:resestbis} are equivalent.
%
\end{theorem}
%


\ni As already mentioned, for precise definition of the set $\FF_j$, see Definition~\ref{def:F(alpha)} and Section~\ref{proof17}. The proof of Theorem~\ref{theo:resest} crucially relies on the propagation result in Theorem~\ref{theo:propagation} below. 
\begin{rem}\label{special}
The special condition at the crossing from~\cite{J6} requires the vanishing of the functions~$(x,\xi )\mapsto \xi\cdot \nabla \Pi _j(x)$
on points where the Hamilton field~$H_j$ is tangent to~$\CC$. Here
we assume the same vanishing on the much smaller set~$\FF_j$. A typical situation where $\FF_j$ is \textit{not} empty is when $\Sigma_j$ contains a piece of a trajectory of the Hamiltonian field $H_j$. 
In Section~\ref{S:degenerated}, we produce an example of this kind for which 
Theorem~\ref{theo:propagation} does not apply, its conclusion is even false, and the vanishing condition of Theorem~\ref{theo:resest} is not satisfied. This strongly suggests 
that Theorem~\ref{theo:resest} does not holds true if this vanishing condition is removed. 
\end{rem}
\begin{rem} Thanks to \eqref{hyp:lr}, there exists $\lambda_0>0$ such that the function 
$(x,\xi)\mapsto x\cdot\xi \cdot {\rm Id}$ is a global escape function at all energy $\lambda\in (\lambda_0,+\infty)$ for the matrix-valued symbol $p:(x,\xi)\mapsto |\xi |^2\cdot {\rm Id}+M(x)$ in the sense of \cite{J2}. 
By Theorem 2.3 in \cite{J2} (which actually holds true with the same proof for all matricial dimension $N$), 
we get (ii) of Theorem~\ref{theo:resestbis} for $I_0=(\lambda_0 ;+\infty)$. Assume now that the assumptions of Theorem~\ref{theo:resest} are satisfied. We derive from Theorem~\ref{theo:resest} that all energy $\lambda\in (\lambda_0,+\infty)$ is non trapping for all fields $H_j$. Then, as in \cite{FR}, one can upperbound the resolvent in \eqref{resest2} by $C\eps^{-1}\lambda^{-1/2}$, where $C$ only depends on $\lambda_0$. Arguing as in \cite{FR}, one gets local in time $H^s$ estimates, smoothing effect and Strichartz estimates,
and one can prove existence and uniqueness of
solutions of non-linear semi-classical Schr\"odinger equation with
matrix-valued potentials in a situation where the potential do not
decrease at infinity.
\end{rem}


\subsection{Codimension~1 crossings and Wigner measures}

\ni A key argument in the proof of Theorem \ref{theo:resest} is a result on propagation of Wigner measures in
presence of degenerated codimension $1$ crossing.
We next present this result, which is of interest by itself. We will work in a general pseudodifferential framework,
as in \cite{[CdV]}, \cite{[CdV2]}, \cite{FG03} and \cite{F04}.  This framework contains the one of Theorem \ref{theo:resest}. We refer to \cite[Section 2]{[CdV]} for other applications.

\noindent We first recall a few facts about Wigner measures (see \cite{Ge93}, \cite{GeLe93}, \cite{GMMP}, \cite{LP}
or the survey \cite{CFMS}). Consider~$(\psi^\eps)_{\eps >0}$ a bounded family in the weighted $L^2$-space $L^{2}_{-s}(\R^d;\C^N):=L^{2}(\R^d;\C^N; \langle x\rangle ^{s}dx)$, for some $s\geq 0$.
Then there exists a positive hermitian Radon measure~$\mu$ and a sequence~$\eps_k$ going to~$0$ as~$k$ goes to~$+\infty$ such that
\begin{equation}
\label{def_mu}
\forall a\in {\cal C}_0^\infty(\R^{2d};\C^{N,N}),\quad  \left(\op_{\eps_k}(a) \psi^{\eps_k},\psi^{\eps_k}\right)\td_k,{+\infty}
{\rm tr}\,\int_{\R^{2d} }a(x,\xi)\d\mu(x,\xi).
\end{equation} 
Here~$\op_{\eps}(a)$ denotes the semi-classical Weyl quantization of~$a$, namely the operator defined by
\begin{equation}\label{eps-Weyl}
\op_\eps(a) f(x)=\int_{\R^{2d}} a\left(\frac{x+x'}{2},\xi\right){\rm
  e}^{\frac{i}{\eps}\xi\cdot(x-x')}f(x')\frac{\d x'\,\d\xi}{(2\pi\eps)^d}.
\end{equation}
We recall here that, when $a\in {\cal C}_0^\infty(\R^{2d})$, \eqref{eps-Weyl} defines an operator which is continuous, uniformly with respect to $\eps$, from~$L^{2}_{-s}$ to~$L^2$ and from~$L^{2}_{\rm loc}$ to~$L^{2}_{\rm loc}$. At many places in this paper, we shall use well-known properties of semi-classical pseudodifferential calculus. See \cite{ds,ma} for details.

\ni The matrix-valued measure~$\mu$ describes the oscillation of the sequence $(\psi^{\eps_k})_k$ which are exactly of size~$1/\eps_k$  or smaller.
Such measure is called a \textit{Wigner measure} associated to the family $(\psi^{\eps})_\eps$. It is a positive hermitian matrix-valued measure in the sense that
for all scalar positive smooth compactly supported test-function~$a$, the quantity~$\int a(x,\xi)\d\mu(x,\xi)$
 is a positive hermitian matrix.

\ni Let $m\in \{1,\ldots N\}$ and consider $m$ real-valued smooth
functions $(\lambda_j)_{j=1\ldots m}$, $m$ matrix valued smooth
functions $(\Pi_j)_{j=1\ldots m}$ on $\R^{2d}$.
 Again we assume that the $\Pi_j(x,\xi)$ are orthogonal projectors satisfying $\Pi_j\Pi_k\equiv 0$ if $j\neq k$. Let 
$$Q(x,\xi)=\sum_{j=1}^m \lambda_j(x,\xi) \Pi_j(x,\xi).$$
We will also assume that $Q$ satisfies, for some real $r_1$ and $r_2$, 
\begin{equation}
\label{H_Q}
\forall \alpha ,\beta\in\N^d, \exists C_{\alpha ,\beta}>0;\, \left|\partial_x^{\alpha}\partial_{\xi}^{\beta}Q(x,\xi)\right|\leq C_{\alpha ,\beta}(1+|x|)^{r_1-|\alpha|}(1+|\xi|)^{r_2-|\beta |}.
\end{equation} 
Let us mention that the result and the computations of this subsection are essentially local, so one can probably relax assumption \eqref{H_Q} on $Q$.

\ni We consider a family $(\psi^\eps)_{\eps>0}$ such that for all~$\eps$, $\psi^\eps$ belongs to the domain of~$\op_\eps(Q)$. We assume that the family is bounded in $L^{2}(\R^d ;\C^N)$ and 
satisfies, on some open subset $\Omega$ of $\R^d$, 
\begin{equation}\label{eq:psi}
\op_\eps(Q)\psi^\eps=o(\eps) \text{ in }L^2(\Omega,\C^N), \text{ as }\eps\to 0. 
\end{equation}

\ni We define again the crossing set $\Gamma$ by
\begin{equation}\label{1.8'}
 \Gamma=\bigcup_{j=1}^m \Gamma_j,\quad\Gamma_j=\{(x,\xi)\in T^*\Omega ;\; \exists k\neq j,
\; \lambda_j(x,\xi)=\lambda_k(x,\xi)\}
 \end{equation}
and we assume that for all $j$, $\Gamma_j$ is included in a codimension $1$ submanifold $\Sigma _j$. 

As above, for~$1\leq j\leq m$, we denote by $H_j$ the Hamiltonian vector fields associated with the functions~$\lambda_j$ and by $\Sigma _{j}(\infty ;H_j)$ the set of points~$(x,\xi)$ in~$\Sigma _j$ where~$H_j$ is tangent at infinite order to~$\Sigma _j$. 
We can define the previous closed subset $\FF_j$ of $\Sigma _{j}(\infty ;H_j)$, that contains all the characteristic curves of $H_j$ that are included in $\Sigma _j$. Recall that $\FF_j$ is empty if $\Sigma _{j}(\infty ;H_j)$ is at most countable. Here again, we refer to  Definitions~\ref{def:F(alpha)} and \ref{def:FF} for  a precise definition. \\
\ni Finally, for smooth matrix-valued functions
$a,b$ on $T^\ast\R^d$, the Poisson bracket $\{a,b\}$ is the
matrix-valued function defined by $\nabla _\xi a\cdot \nabla _x
b-\nabla _x a\cdot\nabla _\xi b$. Setting, for all $j$, 
\begin{eqnarray}\label{def:Bj}
B_j & = &- \frac{1}{2} \{\Pi _j,Q+\lambda _j{\rm
  Id}\}\\
  \label{def:Rj}
R_j & = & \bigl[\{\lambda_j,\Pi_j\},\Pi_j\bigr] +\frac{1}{2}\sum_k(\lambda_k-\lambda_j)\Pi_j\{\Pi_k,\Pi_k\}\Pi_j\, ,
  \end{eqnarray}
 our propagation result is the following 
\begin{theorem}\label{theo:propagation}
Let $\mu$ be any Wigner measure of the family~$(\psi^\eps)_{\eps >0}$
and~$j\in\{1,\cdots ,m\}$. Let~$\omega$ be a bounded open neighborhood of some~$(x^*,\xi^*)\in\Omega\times\R^d$ such that $\overline \omega\subset\Omega\times\R^d$.
 Assume that~$B_j$ vanishes on~$\FF_j\cap \omega$. Then as distributions on~$\omega$,
\begin{equation}\label{eq:propa}
H_j\left(\Pi_j\mu\Pi_j\right)\ =\ [R_j,\Pi_j\mu\Pi_j]\, .
\end{equation}
In particular, the trace of $\Pi_j\mu \Pi_j$ is invariant under the flow
of $H_j$.
\end{theorem}
\begin{rem} 
If $(x^*,\xi^*)\not\in\Sigma_j$, then \eqref{eq:propa} holds true on small enough $\omega$ without 
any assumption on $B_j$. 
\end{rem}
\begin{rem} Replacing the operator $\op_\eps(Q)$ by $D_t-\op_\eps(Q)$, one can show that any Wigner measure $\mu$ of
the time-dependent family $(\psi^\eps)_{\eps >0}$ satisfies
$$(\partial_t-H_j)\left(\Pi_j\mu\Pi_j\right)\ =\ [R_j,\Pi_j\mu\Pi_j]\, .$$
\end{rem}
\begin{rem}\label{rem-extension}
Theorem~\ref{theo:propagation} actually implies the same theorem for $(\psi^\eps)_{\eps >0}$ bounded in 
some $L^2_{-s}$ with $s>0$ and satisfying \eqref{eq:psi} with $L^2(\Omega,\C^N)$ replaced by $L^2_{loc}(\Omega,\C^N)$ (see Remark~\ref{rem-preuve-extension}). 
\end{rem}
\begin{rem}\label{remP} If $Q(x,\xi)=\frac{1}{2}|\xi|^2+M(x)$ with $M$ as in~\aref{formeM}, then $\op_\eps(Q)=P(\eps)$ and
$$B_j(x,\xi)=\frac{1}{2}\xi\cdot\nabla\Pi_j(x),\;\;R_j=\bigl[\xi\cdot\nabla \Pi_j(x),\Pi_j(x)\bigr].$$
\end{rem}
\ni Motivated by the comments in Remark~\ref{special}, we also analyse in Section~\ref{S:degenerated} a strongly degenerated situation which is excluded in Theorem~\ref{theo:propagation}. Finally, in the
Appendix, we give a microlocal normal form which should be of
interest for studying at any order in~$\eps$ a solution to a
partial differential equation close to a non-degenerated point in a
codimension~1 crossing.

\subsection{Comments on the results}

\ni The analysis of the propagation of Wigner measures in presence of eigenvalue crossing has been the subject of intensive works in the last ten years. The existing results are usually devoted to {\it generic}  situations where the Hamiltonian vector fields associated with the eigenvalues are transverse to the crossing set (see \cite{F04}--\cite{FG03}). Theorem~\ref{theo:propagation} covers more general situations where the Hamiltonian fields may be tangent to this set. To our knowledge, it is the first result on the propagation of Wigner measure in presence of eigenvalue crossing in a degenerated situation. We point out that G.~Hagedorn gave an important pioneer contribution to such propagation phenomena in \cite{ha}, where he presented a systematic study of the propagation of a
Gaussian wave packet through generic crossings of various codimension. 
We also want to quote the thesis of U.~Karlsson \cite{Kar} for the construction of a parametrix
in the presence of a smooth eigenvalue crossing and the work of M.~Brassart~\cite{Br} who studied codimension~1 eigenvalue crossings in a periodic situation.  \\
\ni Finally, let us mention that, for codimension~2 and~3 crossings, normal forms have been obtained by Y.~Colin de
Verdi\`ere in \cite{[CdV]} and \cite{[CdV2]}. These important results yield a very detailed description of the solution
close to some generic point in the crossing.  We give here a similar normal form for codimension~1 crossing in the Appendix under a non-degeneracy condition consisting in assuming the transversality of the classical trajectories
to the crossing set and the fact that the gap between the eigenvalues vanish at order $1$ on the crossing set.\\

\ni Concerning the resolvent estimates, there are many results for smooth, scalar Schr\"odinger operators (see \cite{Bu,Duy,Ge90,GeMa,J5,rt,VaZw00,w}). For less regular but still scalar potential, we quote \cite{cjk,cj}. 
In the matricial case, there are rather few results since propagation results like Theorem~\ref{theo:propagation} are difficult to obtain. Indeed, it is rather involved to control the influence of eigenvalues crossings. We quote \cite{FR,J1,J2,J6}. On the related question of existence of resonances for matrix Schr\"odinger operators, we mention \cite{fln,ne}, where only $2\times 2$ matrix operators are considered. Notice that the mentioned resolvent estimates are of great interest for 
semiclassical, molecular scattering theory (a theory for chemical reactions), since a matricial 
Schr\"odinger operator is a toy model for real molecules.

$ $

\noindent Let us say a few words about the proofs. The proof of Theorem~\ref{theo:resest} follows the strategy of \cite{J6}. The necessity of
the non trapping condition is proved following the method of
X.P.~Wang in \cite{w}, as adapted in \cite{FR} using Wigner measures (in
particular without the condition at infinity of \cite{J6}). 
We refer to~\cite{Bu2} for a similar proof. The
sufficiency of the non trapping condition is, as in \cite{J5,J6}, obtained by
contradiction using the method of N.~Burq in \cite{Bu}, which is inspired by an argument of G.~Lebeau in~\cite{Le}, 
and also by the use of a rescaled Mourre estimate at space infinity derived in \cite{J6} (see also \cite{J5}).
The idea is to use Wigner measure to show that some particular sequence, which negates the resolvent bound,
tends to $0$ in $L^2_{\rm loc}$. One step of the proof is to show, using the long-range condition \eqref{hyp:lr}, that the Wigner measure is compactly supported. This is done in \cite{J6} and follows from the rescaled Mourre estimate. 
The new ingredient in the proof of Theorem \ref{theo:resest}
is the propagation of the Wigner measures at the crossing set,
which follows from Theorem \ref{theo:propagation}, under much
weaker assumptions than in \cite{J6}. The proof of this result
relies on an induction on the order of tangency of the flow,
reminiscent of an argument due to R.B.~Melrose and J.~Sj\"ostrand
in the context of propagation of singularities of boundary value
problems \cite{MS}. A standard induction, as in \cite{MS}, gives the propagation
around points of the crossing set where the flow has a finite order contact with the crossing set. Here we use a transfinite induction to get the propagation in
a larger set that may also contains some points with infinite order contact. Finally, the study of the degenerated situation in Section
\ref{S:degenerated} relies on the use of a two-scale Wigner
measure (see \cite{Mi,FG}) to describe more precisely the behaviour
of the Wigner measure at the crossing set.

$ $

\ni The organization of the paper is the following. The two steps
of the proof of Theorem~\ref{theo:resest}, which are by now
classical, are quickly sketched in Section 2, assuming that
Theorem~\ref{theo:propagation} is true. The latter is proved in
Section~\ref{sec:prop}. Then, in Section~4, we analyze the strongly 
degenerated situation mentioned above. Finally, in the Appendix, a 
normal form is given in a non-degenerate situation. 

\medskip

\ni\textbf{Acknowledgment.}
The second author wish to thank Patrick G\'erard for fruitful discussions on the subject. The first author wish to thank Thomas Chomette and Gilles Godefroy for clarifications on ordinal numbers.

\section{Non-trapping condition and resolvent estimate}
\label{sect:resolv-esti}

Here we assume Theorem~\ref{theo:propagation} true and we prove
Theorem~\ref{theo:resest}. We work under the assumptions of
Theorem~\ref{theo:resest}. We do not need to understand what is the
set ${\cal F}_j$. Its meaning is relevant for the proof of
Theorem~\ref{theo:propagation} only.

\subsection{The necessity of the non-trapping condition}

\noindent  Let us first focus on the necessity of the non-trapping
condition. We adapt the arguments in~\cite{FR}, which are inspired
by \cite{w}. Let $E\in I_0$ and $\theta\in{\cal C}_0^\infty(I_0)$
with $\theta =1$ near $E$. Let $I$ be the support of $\theta$. The
resolvent estimate \eqref{resest2} implies that, for $s>1/2$ and
uniformly w.r.t. $\eps\in]0,\eps_0[$, the function~$\langle
x\rangle ^{-s}$ is $P(\eps)$-smooth on $I$ (see Theorem~XIII 25 in
\cite{[RS4]}). Therefore there exists a constant~$C_0$ such that
for any $\eps\in]0,\eps_0[$ and any $\psi\in L^2$
\begin{equation}\label{Psmoothness}
\int_\R\,\left\|\langle x\rangle ^{-s}
\theta\left(P(\eps)\right)\e^{-i\frac{t}{\eps}P(\eps)}\psi\right\|^2_{L^2} \,\d t\leq C_0\left\|\psi\right\|_{L^2}^2.
\end{equation}
\ni We are going to prove that if $(x_j(t),\xi_j(t))_{t\in\R}$ is a
classical trajectory of $H_j$ of energy $E$, that is contained in
$\lambda _j^{-1}(E)$,
then
\begin{equation}\label{contradiction}
\forall T>0,\;\;\int_{-T}^{+T} \,\langle x_j(t)\rangle ^{-s}\,\d t\leq C_0.
\end{equation}

\ni This property implies that $E$ is non-trapping (cf. \cite{w}). 
Let us prove~\aref{contradiction}.
Consider a trajectory $(x_j(t),\xi_j(t))_{t\in\R}$ of $H_j$ of energy $E$ and let $(\psi^\eps_0)_{\eps >0}$ be a by one bounded family in $L^2(\R^2,\C^N)$ having only one Wigner measure $\mu_0$ such that
$$\mu_0=c_0\,\delta(x-x_j(0))\otimes\delta(\xi-\xi_j(0))\,\Pi_j(x).$$
One can actually choose coherent states microlocalized at
$(x_j(0),\xi_j(0))$, for instance. We consider $\psi^\eps (t)=\e^{-i\frac{t}{\eps} P(\eps)}\psi^\eps_0$. The family $(\psi^\eps)_{\eps >0}$ satisfies
$$\frac{\eps}{i}\partial_t\psi^\eps+P(\eps) \psi^\eps=0,\;\;\psi^\eps_{|t=0}=\psi^\eps_0.$$
Let $\mu_t$ be a  Wigner measure of $(\psi^\eps(t))_{\eps >0}$. Let $k\neq j$. 
By Theorem~\ref{theo:propagation}, the measure $\Pi_k\mu _t\Pi_k$ satisfies a linear differential equation with 
initial data $\Pi_k\mu _0\Pi_k=0$, thus it is zero. Since $\Pi_k\mu _t\Pi_j$ is absolutely continuous w.r.t. 
$\Pi_k\mu _t\Pi_k$ and $\Pi_j\mu _t\Pi_j$, it is also zero. Thus $\mu_t=\Pi_j\mu _t\Pi_j$. 
Theorem~\ref{theo:propagation} yields
\begin{equation}\label{calculmu}
{\rm tr}\, \mu(t,x,\tau,\xi)=c_0\,\delta(x-x_j(t))\otimes\delta(\xi-\xi_j(t))\otimes
\delta\left(\tau+\lambda _j(x,\xi )\right)\otimes dt.
\end{equation}
\ni Let $p(x,\xi )=\frac{|\xi|^2}{2}+M(x)$ be the symbol of
$P(\eps)$. Take $T>0$ and a non-negative, smooth, compactly
supported, scalar function $(t,x,\xi)\mapsto a(t,x,\xi)$ such that
\begin{equation}\label{condition-sur-a}
\forall t\in[-T,T],\;\;a\bigl(t,x_j(t),\xi_j(t)\bigr)=1\hspace{.3cm}\mbox{and}\hspace{.3cm}(1-\theta (p))a(t,\cdot ,\cdot )=0\, .
\end{equation}
%


\ni By~\aref{Psmoothness} and the fact that $\|\psi^\eps_0\|_{L^2}\leq 1$, for all $\eps$,
\begin{equation}\label{termborne}
\int_{-T}^{+T}\Bigl(\op_\eps(a(t,x ,\xi ))
\langle x\rangle ^{-s}\theta\left(P(\eps)\right) \psi^\eps(t)
\,|\,\langle x\rangle ^{-s} \theta\left(P(\eps)\right)
 \psi^\eps(t)\Bigr)_{L^2}\,\d t\leq C_0.
 \end{equation}
We first observe that, uniformly w.r.t. $t\in[-T;T]$,
$$\langle x\rangle ^{-s}\op_\eps\left(a(t,x ,\xi )\right)\langle x\rangle ^{-s}=\op_\eps\left(\langle x\rangle ^{-2s}a(t,x ,\xi )\right)+o(1)\;\;{\rm
in}\;\;{\cal L}\left(L^2(\R^d)\right).$$
Then, the functional calculus and \eqref{condition-sur-a} give, in $L^2(\R^d)$,
$$\op_\eps\left(\langle x\rangle ^{-2s}a(t,x ,\xi )\right)
\Bigl(1-\theta(P(\eps))\Bigr)=\op_\eps\Bigl(\langle x\rangle ^{-2s}a(t,x ,\xi )\bigl(1-\theta (p(x ,\xi ))\bigr)\Bigr)+o(1)=o(1)\, $$
where by matricial functional calculus for fixed $(x,\xi )$,
$\theta(p(x,\xi ))=\sum_{1\leq j\leq m}\theta(\lambda _j(x,\xi
))\Pi_j(x)$.
Therefore, writing $\psi^\eps(t)  = \theta\left(P(\eps)\right) \psi^\eps(t)
+({\rm Id}-\theta(P(\eps))) \psi^\eps(t)$, 
$$\displaylines{\int_{-T}^{+T}\Bigl(\op_\eps(a(t,x ,\xi ))
\langle x\rangle ^{-s}\psi^\eps(t)
\,|\,\langle x\rangle ^{-s}
 \psi^\eps(t)\Bigr)_{L^2}\,\d t\hfill\cr\hfill
 =\int_{-T}^{+T}\Bigl(\op_\eps(a(t,x ,\xi ))
\langle x\rangle ^{-s}\theta\left(P(\eps)\right) \psi^\eps(t)
\,|\,\langle x\rangle ^{-s} \theta\left(P(\eps)\right)
 \psi^\eps(t)\Bigr)_{L^2}\,\d t+o(1).\cr}$$

\ni Using~\aref{termborne} and passing
to the limit $\eps\rightarrow 0 $, we get, since $a$ is scalar and satisfies \eqref{condition-sur-a}, 
\begin{eqnarray*}
\int_{-T}^{+T}\int a(t,x,\xi)\langle x\rangle ^{-2s} {\rm tr}\, \d \mu_t(x,\xi)\d t &=&
\int_{-T}^{+T} a(t,x_j(t),\xi_j(t))\langle x_j(t)\rangle ^{-2s} \d t\\
&=&\int_{-T}^{+T} \langle x_j(t)\rangle ^{-2s} \d t \leq C_0\, ,
\end{eqnarray*}
whence~\aref{contradiction}.

\subsection{The sufficiency  of the non-trapping condition}

\label{SS:sufficient}

Now we assume that the non trapping condition is fulfilled on some
open interval $I_0$ included in~$(\|M_\infty\|_\infty ;+\infty)$
and we prove the resolvent estimate \eqref{resest2} by
contradiction. Suppose that, for some interval $I\Subset I_0$,
some $s>1/2$, and some $\epsilon _0>0$, \eqref{resest2} is false.
Then it is shown in \cite{J6} that the following situation occurs:
there exist a sequence~$(\epsilon _n)\in (0;\epsilon _0)^\N$
tending to~$0$, a sequence~$(f _n)$ of $H^2(\R^d)$-functions, a
sequence~$(z_n)\in \C^\N$ such that~${\cal R}e\, z_n\to E\in I$,
$({\cal I}m\, z_n)/\epsilon _n\to 0$, $(f _n)$ is bounded in
$L^2_{-s}(\R^d)$, has a unique Wigner measure $\mu$, and 
the $L^2$-norm $\|\langle x\rangle ^s(P(\eps _n)-z_n)f_n\|_{L^2}$ is 
a $o(\eps _n)$. 
Furthermore, the long range condition \eqref{hyp:lr} 
implies the existence of some $R>0$ such that 
\begin{equation}\label{controle-infini}
\lim_{n},{\infty}\int _{|x|\geq R}\langle x\rangle ^{-2s}|f_n(x)|^2\, dx\ =\ 0. 
\end{equation}
This implies in particular that $\mu$ is nonzero and supported in the compact set
\begin{equation}\label{localisation-mesure}
\biggl(\bigcup _{1\leq j\leq m} \lambda_j^{-1}(E)\biggr)\cap \{(x,\xi ); |x|\leq R\}. 
\end{equation}
Taking $\tau\in {\cal C}^\infty_0(\R^d)$ such that $\tau =1$ near the set $\{x\in\R^d;|x|\leq R\}$, 
one can show, by direct computations, that the sequence $(g_n)$, defined by $g_n=\tau f_n$, 
is bounded in $L^2(\R^d)$, has $\mu$ as unique Wigner measure, and satisfies 
$\|(P(\eps _n)-{\cal R}e\, z_n)g_n\|=o(\eps _n)$ (as in \cite{cjk}). 
In view of Remark~\ref{remP}, we can apply Theorem~\ref{theo:propagation} to $(g_n)$. 
Since, for all $j$, the scalar measure ${\rm tr}\Pi _j\mu\Pi _j$ is compactly supported and
invariant under the flow of $H_j$, the non-trapping condition for
$H_j$ imposes that ${\rm tr}\Pi _j\mu\Pi _j=0$. Since the diagonal terms of 
the matricial measure $\Pi _j\mu\Pi _j$ (recall that $\Pi _j$ may have rank $>1$) are 
non negative, they all vanish. Since the off-diagonal 
terms of $\Pi _j\mu\Pi _j$ are absolutely continuous w.r.t. the 
diagonal terms, they also vanish. Thus, $\Pi _j\mu\Pi _j=0$ for all $j$. 
Since the measures $\Pi _k\mu\Pi _j$ with $k\neq j$ are absolutely continuous w.r.t. 
$\Pi _k\mu\Pi _k$ and $\Pi _j\mu\Pi _j$, they all vanish and $\mu=0$, yielding the desired contradiction.

\begin{rem}\label{re-special} 
Under the previously mentioned special condition at the crossing, it
is shown in \cite{J6} that, for all $j$, the matricial measure $\Pi _j\mu\Pi _j$ is
actually invariant under the flow of $H_j$. Here we use a weaker information namely the propagation 
result in Theorem~\ref{theo:propagation}. \\
The derivation of the properties of $(g_n)$ is quite immediate here. In \cite{cjk} it is 
more complicated due to the presence of singularities. 
\end{rem}



\section{Propagation of Wigner measures}\label{sec:prop}

In this section, we prove Theorem \ref{theo:propagation}. First we explain in the next remark why 
the extension of Theorem \ref{theo:propagation} announced in Remark~\ref{rem-extension} holds true. 
\begin{rem}\label{rem-preuve-extension}
Assume that $(\psi^\eps)_{\eps >0}$ bounded in some $L^2_{-s}$ with $s>0$ and satisfy \eqref{eq:psi} with $L^2(\Omega,\C^N)$ replaced by $L^2_{loc}(\Omega,\C^N)$. Given a bounded open subset $\omega$ of $\Omega$ with $\overline{\omega}\subset\Omega$, take functions $\tau _1,\tau_2\in {\cal C}^\infty_0(\Omega )$ such that $\tau _1\tau_2=\tau_2$ and $\tau_2=1$ near 
$\omega$. Let $Q_1$ be defined by $\tau _2op_\eps (Q)=op_\eps (Q_1)$ and $\phi ^\eps=\tau _1\psi ^\eps$. 
Then we can apply Theorem \ref{theo:propagation} for $(\phi ^\eps)_{\eps >0}$ and $Q_1$. This yields 
\eqref{eq:propa} for any Wigner measure of $(\psi^\eps)_{\eps >0}$ on this $\omega$. 
\end{rem}

\ni Without loss of generality, we may assume that $(\psi^\eps)_\eps$ has only one Wigner measure $\mu$. It is a
straightforward consequence of the functional calculus and
equation \aref{eq:psi} that
$$ \supp \mu \subset \{\det(Q(x,\xi))=0\}=\bigcup_{1\leq j\leq N}^{}\{\lambda_j(x,\xi)=0\}.$$
Let us consider $\chi\in {\cal C}_0^\infty(\Omega\times\R^d;\R)$ such
that $\chi=1$ on $\overline\omega$. We set  $\psi_{j}^{\eps}:=
\op_\eps(\chi \Pi_{j})\psi^{\eps}$. Notice that $\Pi_{j}$ may behave badly at infinity and may not belong 
to a good symbolic class. However $\chi \Pi_{j}\in {\cal C}^\infty _0(\R^d)$ and the family $(\psi_{j}^{\eps})_{\eps >0}$ is bounded in $L^2$. The Wigner measure of
$(\psi^\eps_j)_\eps$ is~$\mu_j=\chi^2\Pi_j\mu\Pi_j$.
Denoting by $\mu_{ij}$ the joint measure of $(\psi_j^\eps)_\eps$ and
$(\psi_i^\eps)_\eps$, for a matrix-valued test function $a$ supported in $\omega$,  
\begin{eqnarray}
& \lim_{\eps},0 \left(\op_\eps(a) \psi_{j}^\eps,\psi_{j}^\eps\right)={\rm tr}\left(\int
a(x,\xi)d\mu_j(x,\xi)\right),&\label{def-mu_j}\\
& \lim_{\eps},0 \left(\op_\eps(a) \psi^{\eps}_{i},\psi^{\eps}_{j}\right)={\rm tr}\left(\int a(x,\xi)d\mu_{ij}(x,\xi)\right)&\nonumber
\end{eqnarray}
The first step of the proof consists in proving the following
proposition which is a consequence of the analysis performed in
Section~6 of \cite{GMMP} and in \cite{GMMP2}. For the convenience
of the reader, we give a proof below in order to emphasize the
specific features due to the crossing. We recall that, for all $j$, the set
$\Gamma _j$ is defined in~\eqref{1.8'} and is included in some codimension one 
submanifold $\Sigma _j$. 
\begin{prop}\label{prop:prop}
For all $j\in \{1,\cdots ,N\}$, there exists a measure $\nu_j$ absolutely continuous with respect to $\mu_j$ such that, as distributions on~$\omega$,
\begin{equation}\label{eq:muj} H_{j}\mu_{j}=\left[R_j,\mu_j\right]+\nu_j,\;\;
{\rm Supp}\, \nu_j\subset\Gamma _j\cap\omega
\end{equation}
where $R_j$ is defined in~\aref{def:Rj}.
Besides, in~$\omega$ and  outside $\Gamma _j$, $\mu\Pi_j=\Pi_j\mu=\mu_j.$
\end{prop}

\noindent In the following subsection, we prove
Proposition~\ref{prop:prop}. Then we study a class of matricial equations containing 
\aref{eq:muj} in a second subsection. Finally, the third subsection is devoted
to the conclusion of the proof of Theorem~\ref{theo:propagation}.

\subsection{Proof of Proposition~\ref{prop:prop}}

We work microlocally in $\omega$.
Let us first prove the commutation relation of $\mu$ with $\Pi_j$ outside $\Gamma _j$. This is equivalent
to the fact that for $k\not=j$, $\Pi_k\mu\Pi_j=0$ on $\omega\setminus\Gamma _j$, i.e. to
\begin{equation}\label{orthogonality}
\forall j\in\{1,\cdots, N\},\;\forall k\not=j,\;\; \forall
a\in{\cal
C}_0^\infty(\omega\setminus\Sigma_j,\C^{N,N}),\;\;
\left(\op_\eps(\Pi_k a \Pi_j)\psi^\eps,\psi^\eps\right)\td_\eps,0
0,
\end{equation}
\ni Since $\lambda_j\not=\lambda_k$ on the support of $a$, $\frac{1}{\lambda_k-\lambda_j}\Pi_k a\Pi_j\in {\cal C}_0^\infty$ and 
$$\Pi_k a\Pi_j=\left[Q,\frac{1}{\lambda_k-\lambda_j}\Pi_k a\Pi_j
\right]\, .$$
Therefore, \aref{orthogonality} follows from 
$$\displaylines{\big(\op_\eps(\Pi_k a \Pi_j)\psi^\eps,\psi^\eps\big)=
\left(\left[\op_\eps(Q),\op_\eps\left(\frac{1}{\lambda_k-\lambda_j}\Pi_k a\Pi_j\right)\right]\psi^\eps,\psi^\eps\right)+O(\eps)=O(\eps)\, .\cr}
$$

$ $

\noindent We focus now on the proof of the transport equation \eqref{eq:propa}. Since $\lambda _j$ may 
also not belong to a good symbol class, we introduce a function $\tilde\chi\in {\cal C}_0^\infty(\Omega\times\R^d;\R)$ such that $\chi=\tilde\chi \chi$. By symbolic calculus, 
\begin{eqnarray*}
 \op_\eps (\tilde\chi \lambda _j)\psi_{j}^{\eps}=\op_\eps (\tilde\chi \lambda _j)
\op_\eps (\chi\Pi_j)\psi^\eps&=& \op_\eps (\tilde\chi \lambda _j\chi\Pi_j)\psi^\eps 
+\eps \op_\eps \bigl(\{\tilde\chi \lambda _j,\chi\Pi_j\}/(2i)\bigr)\psi^\eps +o(\eps),\\
o(\eps)=\op_\eps (\chi\Pi_j)\op_\eps (Q)\psi^\eps &=&\op_\eps (\chi\Pi_jQ)\psi^\eps +\eps \op_\eps \bigl(\{\chi\Pi_j,Q\}/(2i)\bigr)\psi^\eps +o(\eps)
\end{eqnarray*}
in $L^2$. Since $\chi\Pi_jQ=\tilde\chi \lambda _j\chi\Pi_j$, this yields 
$$\op_\eps(\tilde\chi \lambda_{j})\psi^{\eps}_j=\frac{\eps}{i}\,\op_\eps(\tilde B_j)\psi^{\eps}+o(\eps),$$
where $$\tilde B_j=\frac{1}{2}\left(-\{\chi \Pi_j, Q\}+\{\tilde\chi\lambda_j,\chi\Pi_j\}\right)=-\frac{1}{2}\{\chi\Pi_j,Q+\tilde\chi\lambda_j{\rm Id}\}.$$
Thus, for $a\in C^{\infty}_0(\omega,\C^{N,N})$,
\begin{eqnarray}\nonumber
\frac{i}{\eps}\left(\left[\op_\eps(a),\op_\eps(\tilde\chi\lambda_j)\right]\psi^{\eps}_{j},
\psi^{\eps}_{j}\right) & = & \frac{i}{\eps}
\left(\op_\eps(a)\op_\eps(\tilde\chi\lambda_{j})\psi^{\eps}_{j},\psi^{\eps}_{j}\right)-\frac{i}{\eps}\left(\op_\eps(a)\psi^{\eps}_{j},
\op_\eps(\tilde\chi\lambda_j)\psi^{\eps}_{j}\right)\\
\label{eq:Hlambda1}
 & = & \left(\op_\eps(a)\op_\eps(\tilde B_j) \psi^{\eps},\psi^{\eps}_{j}\right)+\left(\op_\eps(a) \psi^{\eps}_{j},\op_\eps(\tilde B_j) \psi^{\eps}\right)+o(1)\\
 \label{eq:Hlambda}
 & = & \left(\op_\eps\left(\Pi_ja\tilde B_j+\tilde B_j^*a\Pi_j\right)\psi^\eps,\psi^\eps\right)+o(1),
\end{eqnarray}
where
$\tilde B_j^*=\frac{1}{2}\left(\{Q,\chi \Pi_j\}+\{\tilde\chi\lambda_j,\chi\Pi_j\}\right)=\frac{1}{2}\{Q+\tilde\chi\lambda_j{\rm Id} ,\tilde\chi\Pi_j\}$.
On one hand,
\begin{equation}
\label{onehand}
\lim_{\eps},0
\frac{i}{\eps}\left(\left[\op_\eps(a),\op_\eps(\tilde\chi\lambda_{j})
\right]\psi^{\eps}_{j},\psi^{\eps}_{j}\right)=- {\rm tr}
\int (H_{j} a)(x,\xi) d\mu_j(x,\xi)
={\rm tr}\int a(x,\xi) d(H_{j}\mu_j)(x,\xi),
\end{equation}
since $\tilde\chi=1$ on the support of $a$, and, on the other hand,
\begin{multline}
\label{otherhand}
 \lim_{\eps},0
  \left(\op_\eps\left(\Pi_ja\tilde B_j+\tilde B_j^*a\Pi_j\right)\psi^\eps,\psi^\eps\right)=
{\rm tr}\left(\int(
\Pi_ja\tilde B_j+\tilde B_j^*a\Pi_j) d\mu(x,\xi)\right)
\\
={\rm tr}\left(\int a( \tilde B_j d\mu(x,\xi) \Pi_j+\Pi_j d\mu(x,\xi) \tilde B_j^*)\right)={\rm tr}\int a \,d\tilde\nu_j(x,\xi),
\end{multline}
where $\tilde\nu_j=
\tilde B_j\mu\Pi_j+\Pi_j \mu \tilde B_j^*$. In view of $\chi=1$ on $\overline\omega$, we have in $\omega$,
\begin{equation}\label{tilde-nu-j}
\tilde\nu_j=
B_j\mu\Pi_j+\Pi_j \mu B_j^*
\end{equation}
\ni with $B_j$ defined in~\aref{def:Bj}. Let us study now $\tilde\nu_j$.
This measure describes the limit of the term \aref{eq:Hlambda1}, therefore it is absolutely continuous with 
respect to the Wigner measure of $(\psi^\eps_j)$, namely $\mu_j$. Let us focus on $\tilde\nu_j\,{\bf 1}_{\Gamma _j^c}$, where ${\bf 1}_{\Gamma _j^c}$ denotes the characteristic function of the complement $\Gamma _j^c$ of $\Gamma _j$.
The computation below is done in \cite{GMMP} and \cite{GMMP2}. We
write it for the sake of completeness. In view of the commutation
relation $\mu\Pi_j=\Pi_j\mu$ outside $\Gamma _j$, we  have
$$\tilde\nu_j\,{\bf 1}_{\Gamma _j^c}=\left(B_j\Pi_j\mu+\mu\Pi_j B_j^*\right)\,{\bf 1}_{\Gamma _j^c}.$$
Besides, writing $Q=\sum_{1\leq k\leq m}\lambda_k\Pi_k$, \aref{def:Bj} implies that 
$$\displaylines {2B_j=\{\lambda_j,\Pi_j\}-\sum_{1\leq k\leq m}\Bigl(\lambda_k\{\Pi_j,\Pi_k\}+\{\Pi_j,\lambda_k\}\Pi_k\Bigr),\cr
2B_j^*= \{\lambda_j,\Pi_j\}+
\sum_{1\leq k\leq m}\Bigl(\lambda_k\{\Pi_k,\Pi_j\}+\Pi_k\{\lambda_k,\Pi_j\}\Bigr).\cr}$$
Since $\Pi_j^2=\Pi_j$, $\Pi_j\{\lambda_j,\Pi_j^2-\Pi_j\}\Pi_j=0$ and $\Pi_j\{\lambda_j,\Pi_j\}\Pi_j=0$. Therefore
\begin{equation}\label{pice1}
\Bigl(\{\lambda_j,\Pi_j\}\Pi_j\mu +\mu\Pi_j\{\lambda_j,\Pi_j\}\Bigr)\, {\bf 1}_{\Gamma _j^c}=\Bigl[\bigl[\{\lambda_j,\Pi_j\},\Pi_j\bigr] \,,\,\mu_j\Bigr]\, {\bf 1}_{\Gamma _j^c}.
\end{equation}
Using the general fact that $A\{B,C\}-\{A,B\}C=\{AB,C\}-\{A,BC\}$, we observe that, $k\not=j$,
\begin{eqnarray}
\{\Pi_k,\Pi_j\}&=&\{\Pi_k^2,\Pi_j\}=\Pi_k\{\Pi_k,\Pi_j\}-\{\Pi_k,\Pi_k\}\Pi_j,\label{crochet-kj}\\
-\{\Pi_j,\Pi_k\}&=&-\{\Pi_j,\Pi_k^2\}=\Pi_j\{\Pi_k,\Pi_k\}-\{\Pi_j,\Pi_k\}\Pi_k,\nonumber
\end{eqnarray}
yielding 
\begin{equation}\label{pice2}
-\lambda_k\Bigl(\{\Pi_j,\Pi_k\}\Pi_j\mu-\mu\Pi_j\{\Pi_k,\Pi_j\}\Bigr)\, {\bf 1}_{\Gamma _j^c}=\lambda_k\Bigl[\Pi_j\{\Pi_k,\Pi_k\}\Pi_j,\mu_j\Bigr]\, {\bf 1}_{\Gamma _j^c}.
\end{equation}
The previous bracket identity with $A=B=C=\Pi_j$ implies that $\Pi_j\{\Pi_j,\Pi_j\}=\{\Pi_j,\Pi_j\}\Pi_j$, whence
\begin{equation}\label{pi_j-bracket}
 \Pi_j\{\Pi_j,\Pi_j\}=\Pi_j^2\{\Pi_j,\Pi_j\}=\Pi_j\{\Pi_j,\Pi_j\}\Pi_j,\;\;
\{\Pi_j,\Pi_j\}\Pi_j=\{\Pi_j,\Pi_j\}\Pi_j^2=\Pi_j\{\Pi_j,\Pi_j\}\Pi_j.
\end{equation}
Replacing both $\Pi_j$'s inside the bracket by ${\rm Id}-\sum_{k\not=j}\Pi_k$ and using \eqref{crochet-kj}
with $k\neq j$ replaced by $k\neq l$, we obtain
\begin{equation}\label{pi_j-bracket-pi_j}
\Pi_j\{\Pi_j,\Pi_j\}\Pi_j=\sum_{k\not=j}\Pi_j\{\Pi_k,\Pi_k\}\Pi_j.
\end{equation}
By \eqref{pi_j-bracket} and \eqref{pi_j-bracket-pi_j}, 
\begin{equation}\label{pice3}
-\lambda_j\Bigl(\{\Pi_j,\Pi_j\}\Pi_j\mu-\mu \Pi_j\{\Pi_j,\Pi_j\}\Bigr)\, {\bf 1}_{\Gamma _j^c}=
-\lambda_j\sum_{k\not=j}\Bigl[\Pi_j\{\Pi_k,\Pi_k\}\Pi_j,\mu_j\Bigr]\, {\bf 1}_{\Gamma _j^c}.
\end{equation}
Collecting the different pieces \aref{pice1}, \aref{pice2} and \aref{pice3}, we obtain
$$2\tilde\nu_j\,{\bf 1}_{\Gamma _j^c} =\biggl[2\bigl[\{\lambda_j,\Pi_j\},\Pi_j\bigr] +\sum_k(\lambda_k-\lambda_j)
\Pi_j\{\Pi_k,\Pi_k\}\Pi_j\;,\;\mu_j\biggr]\,{\bf 1}_{\Gamma _j^c}.$$

\ni By \eqref{onehand} and \eqref{otherhand}, $H_j\mu_j=\tilde \nu_j $ that is $H_j\mu_j=[R_j,\mu_j]+\nu_j$, where $R_j$ is defined by \eqref{def:Rj} and
\begin{equation}\label{defnuj}
\nu_j=\tilde \nu_j \,{\bf 1}_{\Gamma _j}. 
\end{equation}

\subsection{Analysis of a transport equation with a measure valued source term}

We work in a larger setting than those of Theorem~\ref{theo:propagation}. 
We consider an open subset $\Omega$ of $\R^D$ ($D\geq 1$), $\Sigma$ a smooth hypersurface
of $\Omega$, and
\begin{equation}\label{def-H}
H=\sum_{j=1}^Da_j\frac{\partial}{\partial x_j}=:a\cdot \nabla
\end{equation}
a $C^{\infty}$ vector field on $\Omega$  with no singular point
(i.e. $\forall x\in \Omega, H(x)\neq 0$) and with real valued coefficients $a_j$.
Let $\tilde x\in \Sigma$ and $\omega$ be an open neighborhood of $\tilde x$ in $\Omega$ such that $\Sigma\cap \omega=\{\gamma=0\}$,
where $\gamma$ is smooth on $\omega$ and $\nabla \gamma$ does not vanish on $\Sigma\cap \omega$. We start with a few definitions.
\begin{itemize}
\item $\tilde{x}\in \Sigma (0;H)$ if and only if $H$ is transverse to $\Sigma$ at $\tilde{x}$, that is $H \gamma(\tilde{x})\neq 0$.
 \item $\tilde{x}\in \Sigma (k;H)$, $k\in \N^*$ if and only 
if $H\gamma(\tilde{x})=\ldots=H^k\gamma(\tilde{x})=0$ and $H^{k+1}\gamma(\tilde{x})\neq 0$.
\item $\tilde{x}\in \Sigma (\infty;H)$ if and only if $H^k(\tilde{x})=0$ for all $k\in \N$.
\end{itemize}
In the two first cases, we say that $\Sigma$ has a finite order contact with $H$. 
In the sequel we shall need some notions and results of set theory, in particular the notions of ordinals and of transfinite induction. We refer to \cite{Kri,Kri2,Vaught} for details. 
\begin{defi}\label{def:F(alpha)}
\ni Define by transfinite recursion the set $F(\alpha ;H;\Sigma )$, where $\alpha$ is an ordinal, by:
\begin{itemize}
\item $F(0;H;\Sigma )=\Sigma (\infty ;H)$.  
\item If $\alpha=\beta+1$ is a successor ordinal, $F(\alpha ;H;\Sigma )$ is obtained by taking off from
$F(\beta ;H;\Sigma )$ the open subset of points $\tilde{x}\in F(\beta ;H;\Sigma )$ for which there exists an open neighborhood $U$ of $\tilde{x}$ and
an hypersurface $\Theta$ such that $(U\cap F(\beta ;H;\Sigma ))\subset(U\cap \Theta)$ and such that $H$ has only contact of finite order with $\Theta$ at $\tilde{x}$ (i.e. $\tilde{x}\in \Theta (k;H)$ for some $k\in \N$). 
\item If $\alpha=\bigcup_{\beta<\alpha} \beta$ is a limit ordinal, then $F(\alpha ;H;\Sigma )=\bigcap_{\beta<\alpha}F(\beta ;H;\Sigma )$.
\end{itemize}
\end{defi}
\ni Notice that $F(0;H;\Sigma )$, and thus any $F(\alpha ;H;\Sigma )$, is a closed subset of $\Omega$. Furthermore, 
$F(\cdot ;H;\Sigma )$ is non increasing: $F(\beta ;H;\Sigma )\subset F(\alpha ;H;\Sigma )$ if $\beta>\alpha$. 
The family of all $F(\alpha;H;\Sigma)$ is a family of closed subsets of $\R^D$ which is well-ordered for the inclusion. Using that $\R^D$ has a countable basis consisting of open subsets one may show that this family is countable, and thus may be indexed by a countable ordinal. As a consequence, there exists a (countable) ordinal $\beta_0$ such that $F(\alpha ;H;\Sigma )=F(\beta_0 ;H;\Sigma)$, for all $\alpha >\beta_0$.
\begin{defi}\label{def:FF}
We denote by $\FF (H;\Sigma )=F(\beta_0 ;H;\Sigma )$. In the context of Theorems \ref{theo:resest} and \ref{theo:propagation}, we will write,
for $j=1,\ldots ,m$, $\FF_j=\FF (H_j;\Sigma _j)$, where $\Sigma _j$ is an hypersurface containing the eigenvalues crossing and $H_j$ is the Hamiltonian flow associated to the function $\lambda_j$.
\end{defi}

\ni In this section, we prove 
\begin{prop}
\label{P:invariance}
Let $\mu$ be a matrix-valued Radon measure on $\Omega$, such that, as distributions,
\begin{equation}
\label{eq_mu}
H\mu+b\mu+\mu c=\nu,
\end{equation}
where $b$ and $c$ are smooth matrix-valued functions on $\Omega$
and $\nu$ is a matrix-valued Borel measure on $\Omega$ absolutely continuous
w.r.t. $\mu$. Assume furthermore that
\begin{equation}
\label{support_nu}
\supp \nu \subset \Sigma.
\end{equation}
Then $\supp \nu \subset \FF (H;\Sigma )$.
\end{prop}
\begin{rem}
Proposition~\ref{P:invariance} implies that, in the sense of distributions on the open set $\Omega\backslash \FF (H;\Sigma )$,
$$H\mu+b\mu+\mu c=0 .$$
\end{rem}
\ni The proof of Proposition~\ref{P:invariance} proceeds in three steps. In Lemma~\ref{L:invariance}, we prove that ${\bf 1}_{\Sigma (0;H)}\nu =0$,
then, in Lemma~\ref{L:invariance2}, that ${\bf 1}_{\Sigma (k;H)}\nu =0$ for all $k\in\N$. These two facts imply that
${\rm Supp}\,\nu \subset F(0;H;\Sigma )$ and we can conclude the proof by transfinite induction.

\begin{lem}
\label{L:invariance}
Let $H$, $b$, $c$ and $\Sigma$ be as in Proposition \ref{P:invariance}. Then ${\bf 1}_{\Sigma (0;H)}\mu={\bf 1}_{\Sigma (0;H)}\nu=0$.
\end{lem}

\begin{proof} This is a classical property; we give a proof for the sake of
completeness. \\
Since $\nu$ is absolutely continuous w.r.t. $\mu$, it suffices to show ${\bf 1}_{\Sigma (0;H)}\mu=0$.
We must show a local property, namely that if $\tilde x\in
\Sigma (0;H)$, there exists an open neighborhood $U$ of $\tilde x$
such that ${\bf 1}_{U\cap \Sigma}\mu =0$. Let $\tilde x\in \Sigma (0;H)$.

\medskip

\noindent \emph{Step 1. Reduction of the problem.}
We first show that we may assume (locally) that $H=\frac{\partial}{\partial x_1}$, i.e. that we may assume that near $\tilde x$,
$\mu$ satisfies an equation of the form
\begin{equation}
\label{eq_mu_bis}
\frac{\partial}{\partial x_1}\mu+b\mu+\mu c=\nu.
\end{equation}

\noindent Consider a neighborhood $U$ of $\tilde x$ such that $\Sigma$ is defined, in $U$, by the equation $\gamma=0$, where
$\gamma$ is a $C^{\infty}$ function such that $\nabla \gamma(x)\neq 0$ for all $x$ in $U$.
Using that $H$ is transverse to $\Sigma$ in $\tilde x$, and taking a smaller neighborhood $U$ of $\tilde x$ if necessary, one may find
a $C^{\infty}$ diffeomorphism $\chi$ from an open subset $V$ of $\R^D_y$ to $U$ such that
$$ \forall f\in C^{\infty}_0(U),\quad \frac{\partial}{\partial y_1}\left(f\circ \chi\right)=\left(H f\right)\circ \chi.$$
Define the measures $\mu_1$ and $\nu_1$ in $V$ by
$$\forall f\in C^{0}_c(U), \quad \big \langle \mu_1, f\circ  \chi\big\rangle =\big\langle  \mu, f\big\rangle ,\quad \big \langle \nu_1,
f\circ  \chi\big\rangle =\big\langle  \nu, f\big\rangle .$$
By \eqref{eq_mu} and \eqref{def-H}, for $f\in C^{0}_c(U)$, 
\begin{eqnarray*}
 \big\langle  \frac{\partial}{\partial y_1}\mu_1,f\circ \chi\big\rangle &=&-\big\langle \mu_1,
\frac{\partial }{\partial y_1}(f\circ\chi) \big\rangle =
 -\big\langle \mu_1,\left( H f\right)\circ \chi\big\rangle =-\big\langle \mu,H f\big\rangle\\
&=&
\big\langle H\mu,f\big\rangle +\big\langle (\nabla \cdot a)\mu,f\big\rangle \\
&=&\big\langle \nu ,f\big\rangle-\big\langle \mu, bf\big\rangle-\big\langle \mu , fc\big\rangle+\big\langle \mu,(\nabla \cdot a)f\big\rangle\\
&=&\big\langle \nu _1,f\circ\chi \big\rangle-\big\langle (b\circ\chi -(\nabla \cdot a)\circ\chi)\mu _1, f\circ\chi\big\rangle-\big\langle \mu_1(c\circ\chi ), f\circ\chi\big\rangle .
\end{eqnarray*}
Thus $(\partial /\partial y_1)\mu_1+b_1\mu_1+\mu_1c_1=\nu_1$ in $V$, with smooth $b_1:=b\circ \chi -(\nabla \cdot a)\circ\chi$ and $c_1:=c\circ\chi$.

\medskip

\noindent\emph{Step 2. Proof of the vanishing of ${\bf 1}_{\Sigma}\mu$ near $\tilde x$.}
We now assume that \aref{eq_mu_bis} holds near $\tilde x\in \Sigma (0;H)$, where $H=\frac{\partial}{\partial x_1}$. Of course one may assume $\tilde x=0$. Let $U$ a small open neighborhood of $0$ such 
that $\Sigma \cap U=\{\gamma(x)=0\}$ where $\gamma$ is real and smooth, and $\partial_{x_1}\gamma\geq c_0>0$ on $U$. Consider $\chi\in{\cal C}_0^\infty(\R ;\R)$ a bounded function such that $\chi'(0)=1$. 
Let $\psi\in{\cal C}_0^\infty(U)$ with matrix values. Then $\phi :=(1/\partial_{x_1}\gamma)\psi\in{\cal C}_0^\infty(U)$. By \aref{eq_mu_bis}, for all $\eps>0$, 
$$\big\langle \partial_{x_1}\mu ,\chi (\gamma /\eps )\phi \big\rangle =\big\langle \nu ,\chi (\gamma /\eps )\phi \big\rangle -\big\langle \mu ,\chi (\gamma /\eps )b\phi \big\rangle 
-\big\langle \mu ,\chi (\gamma /\eps )\phi c\big\rangle = O(\eps ^0). $$
Therefore 
$$O(\eps ^0)=\big\langle \partial_{x_1}\mu ,\chi (\gamma /\eps )\phi \big\rangle =-\big\langle \mu ,\chi (\gamma /\eps )\partial_{x_1}\phi \big\rangle 
-\frac{1}{\eps}\big\langle \mu ,\chi '(\gamma /\eps )(\partial_{x_1}\gamma)\phi \big\rangle=O(\eps^0). $$
Multiplying by $\eps$ and letting $\eps$ go to $0$, we arrive at: $\big\langle \mu ,(\partial_{x_1}\gamma){\bf 1}_\Sigma\phi \big\rangle =0$. Thus $\big\langle \mu ,{\bf 1}_\Sigma\psi \big\rangle =0$ and ${\bf 1}_\Sigma\mu$ vanishes on $U$. \qed
\end{proof}

$ $

\noindent We next show that we can extend the preceding Lemma to
finite order tangency.
\begin{lem}
\label{L:invariance2}
Let $H$, $b$, $c$ and $\Sigma$ be as in Proposition \ref{P:invariance}. For all $k\geq 0$,
${\bf 1}_{\Sigma (k;H)}\mu={\bf 1}_{\Sigma (k;H)}\nu=0$.
\end{lem}

\begin{proof}
We will use, in a simpler setting, an induction argument on the order at which $H$ is tangent to $\Sigma$ which is related to the argument in \cite[Section 4]{MS} (see also a similar argument in the context of defect measures in \cite{BuLe}).\\
As $\nu$ is absolutely continuous w.r.t. $\mu$, it is sufficient to show
${\bf 1}_{\Sigma (k;H)}\mu=0$, for all $k\geq 0$. 
Let us prove this assertion by induction. By Lemma
\ref{L:invariance}, the property holds for $k=0$. Assume that, for some $k\geq 0$,
\begin{equation}
\label{induction}
{\bf 1}_{\Sigma (0;H)}\mu ={\bf 1}_{\Sigma (1;H)}\mu =\ldots={\bf 1}_{\Sigma (k;H)}\mu=0.
\end{equation}
Let $\tilde x\in \Sigma (k+1;H)$. By the definition of $\Sigma (k+1;H)$, $H^{k+1} \gamma(\tilde x)=0$ and $H^{k+2}\gamma(\tilde x)\neq 0$. 
Let $U$ be a small neighborhood of $\tilde x$ such that
\begin{equation}
\label{def_U}
\forall x\in U, \quad H^{k+2}\gamma(x)\neq 0.
\end{equation}
By \aref{def_U}, the set
$$ \Xi :=\{x\in U,\; H^{k+1}\gamma(x)=0\}$$
is an hypersurface and $H$ is transverse to $\Xi $. The support of $\nu$ is included in
$\Sigma=\bigcup_{0\leq j\leq \infty}\Sigma (j;H)$. Furthermore, $\nu$ is absolutely continuous with respect to $\mu$, thus \aref{induction} implies that ${\rm supp}\nu \subset \bigcup_{k+1\leq j\leq \infty}\Sigma (j;H)$. In particular, ${\rm supp}\nu _{\restriction U}\subset \Xi $. By Lemma \ref{L:invariance} with $\Sigma$ replaced by $\Xi $ and $\Omega$ replaced by $U$,
$${\bf 1}_{\Sigma (k+1;H)\cap U}\mu={\bf 1}_{\{x\in U\;|\; H^{j+1}\gamma(x)=0,\; H^{j+2}\gamma(x)\neq 0\}}\mu =0,$$
which completes the induction argument.
\qed
\end{proof}

$ $

\noindent We can now close the proof of Proposition \ref{P:invariance}.

$ $

\begin{proof}
\ni We will show by transfinite induction that for all ordinal $\alpha$, $\supp \nu \subset F(\alpha ;H;\Sigma )$. For simplicity, denote $F(\alpha ;H;\Sigma )$ by $F(\alpha)$. 
By Lemma \ref{L:invariance2}, $\nu$ vanishes on $\Sigma (k;H)=0$, for all $k$. Using that $\supp\nu \subset \Sigma$ and that
$$ \Sigma \setminus \bigcup_{0\leq k<\infty } \Sigma (k;H)=\Sigma (\infty ;H)=F(0)$$
is a closed set, we get that $\supp \nu \subset F(0)$.

\ni Let us take an ordinal $\alpha>0$ and assume that $\supp \nu \subset F(\beta )$ for all $\beta<\alpha$. If $\alpha$ is a limit ordinal, then by the induction hypothesis $\supp \nu \subset \bigcap_{\beta<\alpha} F(\beta )=F(\alpha )$.

\ni If $\alpha=\beta+1$ for some $\beta$, then we know by the induction hypothesis that $\supp \nu \subset F(\beta )$.
Let $x\in F(\beta )\setminus F(\alpha )$. This means that there exists an open neighborhood $U$ of $x$ in $\Omega$
such that $F(\beta )\cap U\subset \Theta$, where $\Theta$ is a codimension $1$ submanifold of $U$ which has only finite
order contact with $H$. By induction hypothesis, $\supp\nu _{\restriction U}\subset \Theta$. By 
Lemma~\ref{L:invariance2}, 
${\bf 1}_{\Theta}\nu _{\restriction U}=0$, which shows that $x\notin \supp \nu$. Thus $\supp \nu\subset F(\alpha )$, which ends
the proof of Proposition \ref{P:invariance}.
\qed
\end{proof}

$ $

\ni Let us finally  have a closer look at  the set $\FF (H;\Sigma )$. We first notice that if
$\Sigma$ contains a piece $(x(s))_{s\in (0,s_0)}$ of a characteristic curve of the field $H$, then so does $\Sigma (\alpha ;H)$ for all ordinal $\alpha$,
and thus $\FF (H;\Sigma )$ is not empty. A sufficient condition
for $\FF (H;\Sigma )$ to be empty is given by the following proposition, which we will prove by a classical argument using Baire's Theorem.

\begin{prop}
\label{P:vide}
Let $\Sigma$ be a codimension $1$ submanifold in an open subset $\Omega$ of $\R^D$ and $H$ a  smooth vector field in $\Omega$. Assume
\begin{equation}
\label{inclusion}
\Sigma (\infty ;H)\subset \bigcup_{j \in J} Z_j,
\end{equation}
where $\{Z_j\}_{j\in J}$ is a finite or countable family of hypersurfaces that have only a finite order contact with $H$. Then $\FF (H;\Sigma )=\emptyset$. \\
In particular, if $\Sigma (\infty ;H)$ is at most countable or is an hypersurface with only finite order contact with $H$, then $\FF (H;\Sigma )$ is empty.
\end{prop}
\begin{proof}
Denote again $F(\alpha ;H;\Sigma )$ by $F(\alpha)$ for simplicity. 
We will show by induction on the ordinal $\alpha$ that the property $\PP(\alpha)$:
\begin{quote}
``For all ordinal $\beta$ such that $\beta<\alpha$ and $F(\beta )\not=\emptyset$, $F(\alpha )$ is strictly included in~$F(\beta )$.''
\end{quote}
holds for all ordinal $\alpha$.
It is classical that this implies that $F(\alpha )$ is empty if $\alpha$ is large enough.

\smallskip

\ni The property $\PP(0)$ holds trivially. \\
\ni Let $\alpha>0$ be an ordinal and assume that $\PP(\gamma)$ holds for any $\gamma<\alpha$.
First assume that $\alpha$ is a limit ordinal. Let $\beta<\alpha$ and assume that $F(\beta )$ is not empty.
Then $\beta+1<\alpha$, and the property $\PP(\beta+1)$ implies that $F(\beta+1 )$ is strictly included in $F(\beta )$.
As a consequence, $F(\alpha )$ is strictly included  in $F(\beta )$, which shows $\PP(\alpha)$.

\ni It remains to treat the case when $\alpha=\beta+1$ is a successor ordinal. Assume that $F(\alpha )=F(\beta )$.
Using assumption \eqref{inclusion} and that $F(\beta )\subset F(\alpha )$, we get
\begin{equation}
\label{union}
F(\beta )=\bigcup_{j\in J} (F(\beta )\cap Z_j).
\end{equation}
\ni The sets $F(\beta )\cap Z_j$ are closed in $F(\beta )$. We will show that they have empty interior in $F(\beta )$.
If not, there would exist an open subset $U$ of $\Omega$ such that
$\emptyset\neq F(\beta )\cap U\subset F(\beta )\cap Z_j$. As $Z_j$ is an hypersurface which have only finite order contact with $H$,
we would get by the definition of $F(\alpha )=F(\beta+1 )$ that $F(\alpha )\cap U$ is empty, contradicting the fact that
$F(\alpha )=F(\beta )$.

\ni The union \eqref{union} is thus a countable union of closed subsets of $F(\beta )$ that have empty interior in $F(\beta )$.
By Baire Theorem in the complete metric space $F(\beta )$,   we get that $F(\beta )$ has empty interior in itself,
which shows of course that  $F(\beta )$ is empty, concluding the proof of the induction argument and thus of Proposition \ref{P:vide}.\qed
\end{proof}

\subsection{Proof of the propagation }\label{proof17}

Let us now conclude the proof of Theorem~\ref{theo:propagation}. Recall that $\mu_j=\Pi_j\mu\Pi_j$ on 
$\omega$ (see before \eqref{def-mu_j}). By Proposition~\ref{prop:prop}, $H_{j}\mu_j=[R_j,\mu_j]+\nu_j$ on 
$\omega$, where $\nu_j$ is given by~\aref{defnuj} and~\aref{tilde-nu-j}, and ${\rm supp}\, \nu_j\subset\Gamma _j\cap\omega$. Since $\Gamma _j$ is included in a codimension one submanifold $\Sigma _j$, Proposition~\ref{P:invariance} shows that $\supp\nu _j\subset \omega\cap\FF_j$, where $\FF_j=\FF (H_j;\Sigma _j)$ (cf. Definition~\ref{def:FF}). Since the matrix $B_j$ vanishes on $\omega\cap\FF_j$, by assumption, 
so does also $\nu_j$. Thus $\nu_j=0$. This closes the proof  of Theorem~\ref{theo:propagation}.


\section{Highly degenerated situations}
\label{S:degenerated}

In this section, we give an example of $2\times 2$-matrix valued potential $M$ as 
in \eqref{formeM} for which the assumptions of Theorem~\ref{theo:resestbis} are not satisfied. 
Furthermore the conclusion of Theorem~\ref{theo:propagation} (and therefore its hypothesis) 
is also false in that case.

$ $

\noindent{\bf The potential}: We choose $d=N=2$ and take $d_1:\R\rightarrow \R$ a smooth function such that $d'_1(0)>0$. Take $\chi :\R\rightarrow \R$ a smooth function such that $\chi =0$ on $[0;1]$ and positive outside. Let  
$$d(x_1,x_2)=\chi(x_2)d_1(x_1)$$ and, for some $k\in\R$, 
$$M(x)\ =\ d(x){\rm Id}\, +\, \exp\left(-\frac{1}{|x_1|}\right)V(x_2)\hspace{.2cm}\mbox{where}\hspace{.2cm}
V(x_2)\ =\ \left(\begin{array} {cc}{\rm cos}(2k\pi x_2) & {\rm sin}(2k\pi x_2) \\
{\rm sin}(2k\pi x_2) & -{\rm cos}(2k\pi x_2)\end{array}\right).$$
The crossing set is ${\cal C}=\{x_1 =0\}$. We
denote by~$\lambda^{\pm}(x)$, $\Pi^{\pm}(x_2)$ the eigenvalues and
eigenprojectors associated with $M$:
\begin{gather}
\notag
\lambda^{\pm}(x)=\chi(x_2)d_1(x_1)\pm \exp\left(-\frac{1}{|x_1|}\right),\\
\label{explicitPi}
\Pi^+(x_2)=\frac 12\left(\begin{array} {cc} 2\cos^2(k\pi x_2) & \sin(2k\pi x_2) \\ \sin(2k\pi x_2) & 2\sin^2(k\pi x_2)\end{array}
\right),\;\Pi^-(x_2)=\frac 12\left(\begin{array} {cc} 2\sin^2(k\pi x_2) & -\sin(2k\pi x_2) \\ -\sin(2k\pi x_2) & 2\cos^2(k\pi x_2)
\end{array}\right).
\end{gather}

\ni We are concerned with the time-dependent semiclassical Schr\"odinger equation
\begin{equation}\label{eq:schroedinger-t}
i\eps\partial_t \psi^\eps=-\frac{\eps^2}{2}\Delta \psi^\eps +M(x)\psi^\eps
\end{equation}
with initial data $\psi^\eps_0$, such that $(\psi^\eps_0)_{\eps >0}$ is bounded in $L^2(\R^2)$.

$ $

\noindent {\bf The classical trajectories}: The trajectories $s\mapsto (x_1^\pm(s),x_2^\pm(s),\xi_1^\pm (s),\xi_2^\pm (s))$ satisfy
$$\displaylines{\dot x_1^\pm(s)=\xi_1^\pm(s),\;\;\dot \xi_1^\pm(s)=-\chi (x_2^\pm(s))d'_1(x_1^\pm(s))\mp 
\frac{x_1^\pm(s)}{|x_1^\pm(s)|^3}{\rm e}^{-1/|x_1^\pm(s)|}\cr
\dot x_2^\pm(s)=\xi_2^\pm(s),\;\;\dot \xi_2^\pm(s)=-\chi'(x_2^\pm(s))d_1(x_1^\pm(s)),\cr}$$
with the convention that $\exp (-1/|t|)t/|t|^3=0$ for $t=0$. 
We choose the initial condition: $x_1^\pm(0)=\xi_1^\pm(0)=x_2^\pm(0)=0$ and $\xi_2^\pm(0)=\eta >0$. 
We see that, on $[0, 1/\eta ]$, $x_1^\pm(s)=0=\xi_1^\pm(s)$, $\xi_2^\pm(s)=\eta$, and 
$x_2^\pm(s)=s\eta$. Thus, on this time interval, the $+$ and $-$ trajectories coincide and stay in the crossing set $\{(x;\xi )\in(\R^2)^2; x_1=0\}$.

$ $

\noindent{\bf The data}: We choose as initial data a family 
$(\psi^\eps_0)_{\eps >0}$ having a unique Wigner measure of the form
$$\mu(x,\xi)=\delta(x_1)\otimes \delta(x_2)\otimes \delta(\xi_1)\otimes \delta(\xi_2-\eta)\left(a^+\Pi^+(0)+a^-\Pi^-(0)\right)$$
with $a^\pm\in\R^+$. The later is localized at the previously chosen initial 
condition for the trajectories. We expect that the behaviour of the solutions 
of \eqref{eq:schroedinger-t} depends on the concentration around the crossing ${\cal C}$ 
of $(\psi^\eps_0)_{\eps >0}$. 
To describe this concentration, we introduce a two-scale Wigner measures. \\
We choose $\alpha\in ]0,1/2[$ and we analyze
\begin{equation}\label{def-I-eps}
I^\eps(a)=\left(\op_\eps\left(a\left(x,\xi,\frac{x_1}{\eps^\alpha}\right)\right)\psi^\eps_0,\psi^\eps_0\right)
\end{equation}
for symbols $a\in{\cal C}^\infty(\R^5)$ satisfying
\begin{itemize}
\item there exists a compact $K$ of $\R^4$ such that, for all $y\in\R $, $a(\cdot,\cdot,y)\in{\cal C}_0^\infty(K)$,
\item there exists $R_0>0$ and two smooth functions $a(x,\xi,+\infty)$ and $a(x,\xi,-\infty)$  such that for $\mid y\mid>R_0$,
$a(x,\xi,y)=a(x,\xi,{\rm sgn}(y)\infty)$.
\end{itemize}
According to \cite{Mi} (see also \cite{FG}), a two-scale Wigner
measure associated to the concentration of $(\psi^\eps_0)_{\eps >0}$ on
$\{x_1=0\}$ at the scale $\eps^\alpha$ is a nonnegative Radon measure $\nu$
on $\R^4\times \overline\R$ where
$\overline\R=(\R\cup\{-\infty,+\infty\})$ such that, for some sequence $(\eps _k)_k$ tending to $0$, $\mbox{lim}_k\, I^{\eps _k}(a)=\langle a,\nu\rangle $. One then recovers
$\mu$ by projection of $\nu$ on $\R^{2d}$ through
\begin{equation}\label{eq:mu=int-nu}
 \mu(x,\xi)=\int_{\overline\R}\nu(x,\xi,dy)
\end{equation}
and one can decompose $\nu$ as
$$\nu(x,\xi,y)=\delta(x_1)\otimes\left( \gamma(x_2,\xi,y){\bf 1}_{\mid y\mid<\infty}+\gamma_\infty(x_2,\xi,y)
{\bf 1}_{\mid y\mid=\infty}\right)+\mu(x,\xi)\otimes\delta\left(y-\frac{x_1}{\mid x_1\mid}\infty\right){\bf 1}_{\{x_1\not=0\}}.$$
We choose $(\psi^\eps_0)_{\eps >0}$ such that its concentration above $\{(x;\xi )\in(\R^2)^2; x_1=0\}$ at the scale $\eps^\alpha$ is given by a unique two-scale Wigner measure
\begin{equation}\label{eq:initial-nu}
 \nu(x,\xi,y)={\bf 1}_{\mid y\mid<+\infty}\; \delta(x_1)\otimes \delta(x_2)\otimes \delta(\xi_1)
\otimes \delta(\xi_2-\eta)
\left(\gamma^+(y)\Pi^+(0)+\gamma^-(y)\Pi^-(0)\right)
\end{equation}
where
$\gamma^\pm$ are finite nonnegative Radon measures on $\R$. Of course, one
has $a^\pm=\int_{\R}\gamma^\pm(dy)$. We emphasize that the two-scale Wigner measure $\nu$ is supported on $\R^{4}\times\R$ so that $\nu(\{-\infty,+\infty\})=0$. The following proposition
describes the Wigner measure of~$\psi^\eps(t)$ for
$t\in[0,1/\eta ]$.
\begin{prop}\label{prop:propa-ex}
There exists a sequence $(\eps_k)_k$ tending to $0$ and a family of nonnegative Radon measure $\mu _t$ such that for all $t\in[0,1/\eta ]$ and for all
$a\in{\cal C}_0^\infty(\R^{2d},\C ^{2,2})$,
\begin{equation}\label{eq:def-mu_t}
 \left(\op_{\eps_k}(a)\psi^{\eps_k}(t),\psi^{\eps_k}(t)\right)\td_{k},{+\infty} {\rm tr}\, \int_{\R^{2d}} 
a(x,\xi) \d\mu _t(\d x,\d\xi).
\end{equation}
Besides,
\begin{equation}\label{eq:express-mu_t}
 \mu _t(x,\xi)=\delta(x_1)\otimes \delta(x_2-t\eta )\otimes \delta(\xi_1)\otimes 
\delta(\xi_2-\eta) \left(a^+\Pi^+(0)+a^-\Pi^-(0)\right).
\end{equation}
\end{prop}
\begin{rem} 
Let $\Omega =\{x\in\R^2 ; |x|<2\}$ and $\theta\in {\cal C}^\infty _0(\R^2)$ such that $\theta =1$ near $\Omega$. 
The family $(\psi^\eps)_{\eps >0}$ satisfies \eqref{eq:psi} with $Q(x;\xi )=(|\xi |^2/2+d(x)){\rm Id}+\theta (x)M(x)$. 
Except the vanishing condition on the $B_j$, the assumptions of Theorem~\ref{theo:propagation} are satisfied. 
\end{rem}
\begin{rem} We deduce from \eqref{eq:express-mu_t} that the measures $\mu_t^\pm=\Pi^\pm\mu _t\Pi^\pm$ satisfy
$$\mu^\pm_t(x,\xi)=\delta(x_1)\otimes \delta(x_2-t\eta )\otimes \delta(\xi_1)\otimes \delta(\xi_2-\eta) \Pi^\pm(t\eta )
\left(a^+\Pi^+(0)+a^-\Pi^-(0)\right)\Pi^\pm (t\eta ).$$
We observe that, depending on the choice of  $k$, it is possible to have energy transfers between the modes.
More precisely, in view of \eqref{explicitPi}, we have the following facts:
\begin{itemize}
\item For $k\in \Z+1/2$, between $t=0$ and $t=1/\eta $, the mass $a^+$ is transfered from the plus mode to the minus one. The conclusion of Theorem~\ref{theo:propagation} is not true. 
\item For $k\notin \Z\cup (\Z+1/2)$, between $t=0$ and $t=1/\eta $, there is some partial transfer from one mode to the other. The conclusion of Theorem~\ref{theo:propagation} is false. 
\item For $k=0$, the conclusion of Theorem~\ref{theo:propagation} holds true since the $B_\pm$ defined in \eqref{def:Bj}
vanish on the crossing. 
\end{itemize}
\end{rem}

$ $

\begin{proof}
The proof of this proposition relies on the analysis of the two-scale Wigner measure associated
with~$\psi^\eps(t)$ which is given by the following lemma.

\begin{lem}
There exists a sequence $(\eps_k)$ and a family of nonnegative Radon
measure~$\nu_t$ such that for all~$t\in[0,1/\eta ]$ and for
all~$a\in{\cal C}_0^\infty(\R^{2d+1},\C ^{2,2})$,
$$\left(\op_{\eps_k}\left(a\left(x,\xi,\frac{x_1}{\eps^\alpha}\right)\right)\psi^{\eps_k}(t),\psi^{\eps_k}(t)\right)\td_{k},{+\infty}
{\rm tr}\, \int_{\R^{2d}\times \R} a(x,\xi,y) \d\nu_t(\d x,\d\xi,\d y).$$
Besides,  $t\mapsto\nu_t$ is continuous and  $\nu_t$ satisfies the
transport equation
\begin{equation}\label{eq:transport-nu}
\partial_t\nu_t+\xi\cdot\nabla_x\nu_t-\nabla _xd\cdot\nabla_\xi\nu_t=0 
\end{equation}
with initial condition $\nu _0=\nu$, given by \eqref{eq:initial-nu}. 
\end{lem}

\ni Solving the transport equation \eqref{eq:transport-nu}, we obtain
$$\nu_t(x,\xi,y)=\delta(x_1)\otimes\delta(x_2-t\eta )\otimes \delta(\xi_1)\otimes \delta(\xi_2-\eta)
\left(\gamma^+(y)\Pi^+(0)+\gamma^-(y)\Pi^-(0)\right).$$
For each $t\in[0,1/\eta ]$, the sequence $(\psi^{\eps_k}(t))_k$ has a unique Wigner measure $\mu _t$ given by \eqref{eq:mu=int-nu} with $\mu ,\nu$ replaced by $\mu _t,\nu _t$ respectively. This yields \eqref{eq:def-mu_t} and~\eqref{eq:express-mu_t}. 
\qed
\end{proof}

$ $

\noindent Let us now prove the Lemma.

\begin{proof} Consider a symbol $a\in{\cal C}_0^\infty(\R^{2d+1},\C ^{2,2})$ and let us analyze the quantity
$$I^\eps_a(t)=\left(\op_\eps\left(a\left(x,\xi,\frac{x_1}{\eps^\alpha}\right)\right)\psi^\eps(t),\psi^\eps(t)\right).$$
By \eqref{eq:schroedinger-t}, 
$$\frac{\d}{\d t} I_{\eps,a}(t)= \left(K_\eps\psi^\eps(t),\psi^\eps(t)\right)\hspace{.2cm}\mbox{with}\hspace{.2cm}K_\eps=\frac{1}{i\eps}\left[\op_\eps\left(a\left(x,\xi,\frac{x_1}{\eps^{\alpha}}\right)\right),-\frac{\eps^2}{2}\Delta +M(x)\right].$$
We use the scaling operator $T$
defined by
$$\forall f\in L^2(\R^2),\;\; Tf(x)=\eps^{\alpha/2}f(x_1\eps^\alpha, x_2).$$
It is unitary on $L^2(\R^2)$. In ${\cal L}\left(L^2(\R^d)\right)$, $K_\eps$ can be written as 
$$\frac{T^*}{i\eps}\left[\op_1\left(a(x_1\eps^\alpha, x_2,\xi_1
\eps^{1-\alpha},\xi_2\eps,x_1)\right)\;,\;-\frac{\eps^{2(1-\alpha)}}{2}\partial_{x_1}^2-\frac{\eps^2}{2}\partial_{x_2}^2+d(x_1\eps^\alpha,x_2)+{\rm
e}^{-1/(\eps^\alpha\mid x_1\mid)}V(x_2)\right]T$$
\begin{eqnarray}
&=&\op_\eps\left(\xi\cdot
(\partial_xa)(x,\xi,x_1\eps^{-\alpha})-\nabla d(x)\cdot(\partial_\xi
a)(x,\xi,x_1\eps^{-\alpha})\right)\label{eq:rescaled-calculus}\\
&&+\eps^{1-2\alpha}T^*\op_1\left(\xi_1\partial_y
a(x_1\eps^\alpha,x_2,\xi_1\eps^{1-\alpha},\xi_2\eps, x_1)\right)T
+o(1). \nonumber
\end{eqnarray}
Here we used that there exists a constant
$C$ such that on the support of $a(x_1\eps^\alpha,
x_2,\xi_1 \eps^{1-\alpha},\xi_2\eps,x_1)$, we have $\mid x_1\mid
\leq C$, therefore $\frac{1}{\eps }{\rm e}^{1/(x_1\eps^\alpha)}$
goes to $0$ uniformly, as $\eps$ goes to $0$.\\
This implies that $\left(K_\eps\psi^\eps(t),\psi^\eps(t)\right)$ is
uniformly bounded. Therefore, by Ascoli theorem, considering a
dense subset of ${\cal C}_0^\infty(\R^{2d+1},\C ^{2,2})$, and then arguing by
diagonal extraction, one can find a sequence $(\eps_k)_k$ such that, for all~$a \in{\cal C}_0^\infty(\R^{2d+1},\C ^{2,2})$, $I_{\eps_k,a}(t)$ has
a limit as $k$ goes to $+\infty$, uniformly w.r.t. $t$. Besides, observing that
$$(K\psi^\eps(t),\psi^\eps(t))\td_\eps, 0 \langle\xi\cdot\partial_x a(x,\xi,y) -\nabla d(x)\cdot \partial_\xi a(x,\xi,y), \nu_t\rangle $$
by \eqref{eq:rescaled-calculus}, we obtain the transport equation \eqref{eq:transport-nu}.\qed
\end{proof}

$ $


\appendix

\section{Appendix: Normal form for non degenerated codimension~$1$ crossing}
\label{A:normal_form}

We present here a normal form in the spirit of \cite{[CdV]}, \cite{[CdV2]} and \cite{Fe03}. We consider
a generic class of non-degenerated crossings in two matricial dimensions.
We take a $\C^{2,2}$-valued symbol $Q$  and decompose it into a scalar part plus a trace
free matrix:
$$Q(x,\xi)=\phi_0(x,\xi){\rm Id}+\begin{pmatrix}\phi_1 (x,\xi)&
\phi_2(x,\xi)+i\phi_3(x,\xi)\\  \phi_2(x,\xi)-i\phi_3(x,\xi) &
-\phi_1(x,\xi)\end{pmatrix}.$$
Here the functions $\phi _0,\cdots ,\phi _3$ are real-valued and ${\cal C}^\infty$. The eigenvalues
are the functions
$$\lambda^\pm(x,\xi)=\phi_0(x,\xi)\pm\sqrt{\phi_1(x,\xi)^2 +\phi_2(x,\xi)^2 +\phi_3(x,\xi)^2  }.$$
The
crossing set is the set $$\Gamma =\{(x,\xi)\in T^*\R^d;\;
\phi_1(x,\xi)=\phi_2(x,\xi)=\phi_3(x,\xi)=0\}.$$ We assume that, near
some point $(x^*,\xi^*)\in \Gamma$, $\Gamma $ is a codimension one submanifold
of $T^*\R^d$ given by some equation $\gamma(x,\xi)=0$ with $\nabla\gamma\not=0$ on $\Gamma\cap\Omega_0$
where $\Omega_0$ is an open neighborhood of $(x^*,\xi^*)$. \\
Following \cite{[CdV]}, \cite{[CdV2]} and \cite{FG03},
we say that this codimension~$1$ crossing is {\it non degenerated} in $(x^*,\xi^*) $ if
$$\left(\{\phi_0,\phi_1\},\{\phi_0,\phi_2\},\{\phi_0,\phi_3\}\right)(x^*,\xi^*)\not=(0,0,0).$$

\begin{prop} If
the codimension~$1$ crossing is non degenerated in $(x^*,\xi^*)$,
then
\begin{equation}\label{eq:gen}
\{\phi_0,\gamma\}(x^*,\xi^*)\not=0.
\end{equation}
and the Hamiltonian vector fields $H_{\lambda^\pm}(x^*,\xi^*)$ are transverse to the crossing set.
\end{prop}

\begin{proof}
The fact that $\gamma=0$ is an equation of the crossing set yields
that there exist a neighborhood $\Omega_1$ of $(x^*,\xi^*)$ and
smooth functions~$u_j(x,\xi)$, $1\leq j\leq 3$ such that
$$\phi_j(x,\xi)=\gamma(x,\xi) u_j(x,\xi),\;\;\forall j\in\{1,2,3\},\;\;\forall (x,\xi)\in\Omega_1.$$
In view of $\gamma(x^*,\xi^*)=0$, we obtain
$$\forall j\in\{1,2,3\},\;\;\{\phi_0,\phi_j\}(x^*,\xi^*)=u_j(x^*,\xi^*) \{\phi_0,\gamma\}(x^*,\xi^*).$$
Therefore, the quantities $u_j(x^*,\xi^*)$ may not be all equal to~$0$ and $\{\phi_0,\gamma\}(x^*,\xi^*)\not=0$.

\ni Besides, the vector fields $H_{\lambda^\pm}$ are transverse to the crossing set if and only if
$ H_{\lambda^\pm}\gamma(x^*,\xi^*)\not=0.$ We observe that
\begin{eqnarray*}
H_{\lambda^\pm}\gamma(x^*,\xi^* ) & = &  \{\lambda^\pm,\gamma\}(x^*,\xi^*)\\
& = & \{\phi_0,\gamma\}(x^*,\xi^*) \pm\left\{\sqrt{\phi_1^2+\phi_2^2+\phi_3^2},\gamma\right\}
(x^*,\xi^*).
\end{eqnarray*}
As $\gamma(x^*,\xi^*)=0$ and $\phi_j=\gamma u_j$, $1\leq
j\leq 3$, we get
$$  H_{\lambda^\pm}\gamma(x^*,\xi^* )
= \{\phi_0,\gamma\}(x^*,\xi^*)\not=0 ,$$
whence the proposition.\qed
\end{proof}

$ $

Before stating the main result,  let us recall some basic facts about canonical transforms and
 Fourier integral operators.  The phase space $\R^{2d}$ has the structure of the
cotangent  space $T^*\R^{d}$ which is a symplectic space endowed with  the $2$-form
$\omega=\d\xi\wedge\d x$. A canonical transform is a
 local change of symplectic
 coordinates, i.e. a local diffeomorphism which preserves the symplectic structure.
 One can  associate with a canonical transform an operator $U$ called Fourier Integral Operator, which is  a unitary bounded operator  of
 $L^2(\R^{d})$ satisfying convenient properties that we explain now.  The reader will find in~\cite{[Ro]} a complete analysis
of Fourier Integral Operator, the presentation chosen here is the
one of~\cite{FG}. \\

Once given the canonical transform $\kappa$ in some open set
$\Omega$, one constructs a ${\cal C}^1$~path
$\delta\mapsto\kappa(\delta)$, $\delta\in[0,1]$ linking ${\rm Id}$
to $\kappa$ (see \cite[Lemma~1]{FG}). The fact that $\kappa(\delta)$ preserves the
symplectic structure of $T^*\R^{d}$ yields that $\frac{\d}{\d
\delta}\,\kappa(\delta)\circ \kappa(\delta)^{-1}$ is a Hamiltonian
 vector field on $T(T\R^{d})$ above~$\Omega$. Therefore, there exists a smooth
function $f(\delta)$ such that $\kappa(\delta)$ solves
$$\frac{\d}{\d \delta}\,\kappa(\delta)=H_{f(\delta)} \kappa(\delta),\;\;\kappa(0)={\rm
Id},\;\;\kappa(1)=\kappa.$$

\ni Define the family of unitary operators $U^\eps(\delta)$ by 
\begin{equation}
\label{eq:U}
{i\eps}\frac{\d}{\d \delta}
U^\eps(\delta)=\op_\eps(f(\delta))U^\eps(\delta),\;\;U^\eps(0)={\rm
Id} ,\;\;\delta\in[0,1].
\end{equation} 
Then, we are interested in the operator
$U$ defined by
$$U:=U^\eps(1).$$ This operator  satisfies the following formula
known as Egorov's Theorem (see Section 2.2 in \cite{FG})
\begin{equation}\label{eq:Egorov}\forall
a\in{\cal C}^\infty_0(\R^{2d}),\;\;
U^*\op_\eps\left(a\right)U=\op_\eps(a\circ\kappa)+O(\eps^2)\;\;{\rm
in}\;\;{\cal L}(L^2(\R^{d})) . \end{equation} \ni The operator
$U$ is a {\it Fourier Integral Operator} associated with
$\kappa$.

\ni Observe that finding the operator $U$ or finding the function
$f$ are equivalent questions; we will use this fact in the proof
of the  following theorem.

\begin{theorem}\label{fornor}
If the matrix $Q$ presents a non degenerated codimension $1$
crossing in $(x^*,\xi^*)$, then for all $N\in\N$, there exist a
canonical transform $\kappa_N$ from a neighborhood of
$(x^*,\xi^*)$ into a neighborhood $\Omega$ of $0$
$$\kappa_N(x,\xi)=(s,z,\sigma,\zeta),\;\;(s,\sigma)\in\R^2, \;\;(z,\zeta)\in\R^{d-1}\times\R^{d-1}$$
 smooth $\C$-valued  functions $(\gamma_j)_{1\leq j\leq N}$ and  smooth matrices $(B_j(z,\zeta))_{0\leq j\leq N}$
 such that if $U_N$ is a Fourier integral operator associated with $\kappa_N $ then for all $a\in{\cal C}_0^\infty(\Omega)$,
 in ${\cal L}\left(L^2(\R^d)\right)$,
$$\displaylines{\qquad\op_\eps(a)\Bigl(U_N^*\op_\eps(B_{N}^\eps )^*\op_\eps (Q)\op_\eps  (B_{N}^\eps)U_N\Bigr)\hfill\cr\hfill
=\op_\eps(a)\op_\eps\left(\alpha_0\sigma\,{\rm Id}+\begin{pmatrix}
s & \eps\gamma_N^\eps(z,\zeta) \\\eps \ol\gamma_N^\eps(z,\zeta) &
-s\end{pmatrix}\right)+O(\eps^{N+1}),\qquad\cr
B_N^\eps=B_0+\eps B_1+\cdots+\eps^NB_N,\;\;\gamma_N^\eps=\gamma_1+\eps\gamma_2\cdots+\eps^{N-1}\gamma_N,\cr}$$
where $\alpha_0={\rm sgn}\bigl(\{\phi_0,\gamma\}(x^*,\xi^*)\bigr)$
 Besides, $B_0^*B_0=\lambda^{-1}{\rm Id}$ with  $\lambda(x^*,\xi^*)\not=0$.
\end{theorem}

\begin{rem}
This theorem allows to turn microlocally the equation $\op_\eps(Q)\psi^\eps=0$ into
$$\frac{\eps}{i}\partial_s u^\eps= \begin{pmatrix} s  & \Gamma^\eps_N \\ (\Gamma^\eps_N)^* & -s\end{pmatrix}u^\eps+O(\eps^{N+1})$$
where $\Gamma^\eps_N$ is an operator commuting with $s$ and $\partial_s$. The solutions of such systems are described in Proposition~7 of  \cite{FG03}.
\end{rem}

\begin{proof}
We proceed in two steps, following  the strategy of \cite{[CdV]} :
we first work on the classical symbols and find $B_0$ and
$\kappa$, then we work on the operator level and find $B_j$ for
$j\geq 1$. 
Assume that $\{\phi_0,\gamma\}(x^*,\xi^*)>0$. We first
use that $\{\phi_0,\gamma\sqrt{u_1^2+u_2^2+u_3^2}\}(x^*,\xi^*)>0$,
thus by Lemma 21.3.4 in \cite{[Ho]}, there exists a positive
function $\lambda=\lambda(x,\xi)$ such that
$$\left \{\lambda\phi_0,\lambda\gamma\sqrt{u_1^2+u_2^2+u_3^2}\right\}(x,\xi)=1.$$
Then, we set
$$\sigma(x,\xi)=\lambda(x,\xi)\phi_0(x,\xi),\;\;s(x,\xi)=\lambda(x,\xi)\gamma(x,\xi)\sqrt{u_1^2(x,\xi)+u_2^2(x,\xi)+u_3^2(x,\xi)}.$$
We then have
$$\displaylines{\qquad\lambda(x,\xi) Q(x,\xi)\hfill\cr\hfill=\sigma(x,\xi){\rm Id} +\frac{s(x,\xi)}{\sqrt{u_1^2(x,\xi)+u_2^2(x,\xi)+u_3^2(x,\xi)}}\begin{pmatrix} u_1(x,\xi) & u_2(x,\xi)+iu_3(x,\xi)\\ u_2(x,\xi)-iu_3(x,\xi) & -u_1(x,\xi) \end{pmatrix}.\qquad\cr}$$
The matrix $\displaystyle{\frac{1}{\sqrt{u_1^2+u_2^2+u_3^2}}\begin{pmatrix} u_1 & u_2+iu_3\\ u_2-iu_3
& -u_1 \end{pmatrix}}$  smoothly diagonalizes in a neighborhood of
the point $(x^*,\xi^*)$. There exists $B=B(x,\xi)$  smooth and
orthogonal such that
$$\lambda(x,\xi) B^*(x,\xi) Q(x,\xi) B(x,\xi)= \sigma(x,\xi) {\rm Id} + s(x,\xi) \begin{pmatrix} 1 & 0 \\ 0 & -1 \end{pmatrix}.$$
The fact that $\{\sigma,s\}(x,\xi)=1$ yields that by the Darboux theorem, there exists a canonical transform $\kappa_0$ such that $\kappa_0(x,\xi)=(s,z,\sigma,\zeta)$ and
$$\left(\lambda B^* Q B\right)\circ\kappa_0^{-1}(s,z,\sigma,\zeta)=  \sigma {\rm Id} + s \begin{pmatrix} 1 & 0 \\ 0 & -1 \end{pmatrix}.$$
Let us consider now a Fourier integral operator $U_0$ associated with $\kappa_0$. Setting
$$B_0=\sqrt\lambda B,$$
we have
$$\displaylines{\qquad U_0^*\op_\eps(B_0)^*\op_\eps(Q)\op_\eps(B_0)U_0=U_0^*\op_\eps(B_0^*QB_0)U_0+\eps\,\op_\eps \left(R_1 \right)\hfill\cr\hfill
=\op_\eps\left(\sigma {\rm Id} + s \begin{pmatrix} 1 & 0 \\ 0 & -1
\end{pmatrix}\right)+\eps\,\op_\eps(R_1)\qquad\cr}$$ where we have
used symbolic calculus and the definition of $U_0$. It remains now
to get rid of the rest term $R_1$ (which is self-adjoint) by
modifying the canonical transform $\kappa_0$ and the matrix $B_0$.
We proceed by induction: we suppose that we have find $\kappa_N$,
$U_N$, $B_0,\cdots, B_N$ and $\gamma_1,\cdots,\gamma_N$ such that
microlocally near $(0,0)$ we have in ${\cal
L}\left(L^2(\R^d)\right)$
$$U_N^*\op_\eps(B_N^\eps)^* \op_\eps (Q)\op_\eps  (B_N^\eps)U_N
=\op_\eps\left(Q_0+\Gamma^\eps_N\right)+\eps^{N+1}\op_\eps(R_N)+O\left(\eps^{N+2}\right)$$
where
$$Q_0(s,\sigma)=\sigma\,{\rm Id}+\begin{pmatrix} s & 0 \\ 0 &
-s\end{pmatrix},\quad
\Gamma_N^\eps=\begin{pmatrix}0 & \gamma_N^\eps \\
\ol\gamma_N^\eps & 0 \end{pmatrix}.$$  We decompose $R_N $ as
$$R_N = p_N \,{\rm Id} + \begin{pmatrix} d_N & v_N \\ \ol v_N
 & -d_N \end{pmatrix}=\tilde R_N + \begin{pmatrix} 0 & \gamma_{N+1} \\ \ol \gamma_{N+1}
 & 0 \end{pmatrix} ,$$
where $\gamma_{N+1}(z,\zeta)=v_N(0,z,0,\zeta)$, so that for some smooth functions $g$ and  $h$,
$$ v_N(s,z,\sigma,\zeta)=\gamma_{N+1}(z,\zeta)+s g(s,z,\sigma,\zeta) + \sigma h(s,z,\sigma,\zeta).$$
We claim that we can find a function $f$ and a matrix $D$ such that
\begin{equation}\label{eq:homological}
\{f,Q_0\}+D^*Q_0+Q_0D+\tilde R_N=0.
\end{equation}
Let us postpone the proof of this claim at the end of the proof.

$ $

\noindent We then choose
$$B_{N+1}=B_0D,$$
and we  compose $\kappa_N$ with a perturbation of the identity constructed from $f$. We use a family of   canonical transforms
$\chi^\eps(\delta)$ for $\delta\in[0,1]$ which satisfy
$$\frac{\d}{\d\delta}\chi^\eps(\delta)=H_{\eps^{N+1} f}\circ
\chi^\eps(\delta),\;\;\chi^\eps(0)={\rm Id}.$$  If
$U^\eps(\delta)$ is a Fourier integral operator associated with
$\chi^\eps(\delta)$, we have by \aref{eq:U}
$$\frac{\d}{\d\delta}\left({U^\eps(\delta)^*}\op_\eps(a)U^\eps(\delta)\right)=\eps^{N+1}{U^\eps(\delta)^*}\op_\eps\Bigl(\left\{
f,a\right\}\Bigr)U^\eps(\delta)+O(\eps^{N+2})\;\;{\rm in}\;\;{\cal
L}\left(L^2(\R^d)\right).$$ Consider
$$\displaylines{\qquad C^\eps(\delta)=U^\eps(\delta)^*\Bigl[U_N^*\op_\eps\left((B_N^\eps)^*+\delta\eps^{N+1}B_{N+1}^*\right)
\op_\eps(Q)\op_\eps\left(B_N^\eps+\delta\eps^{N+1}B_{N+1}\right)U_N\hfill\cr\hfill
-(1-\delta)\eps^{N+1}\op_\eps(\tilde R_N)
\Bigr]U^\eps(\delta).\qquad\cr}$$ Setting
$$U_{N+1}=U_N\circ U^\eps(1),$$
we have, microlocally near $(0,0)$
$$\displaylines{C^\eps(0)=\op_\eps(Q_0+\Gamma_{N+1}^\eps)+O(\eps^{N+2}),\cr
C^\eps(1)=U_{N+1}^*\op_\eps\left((B_{N+1}^\eps)^*)\op_\eps(Q)\op_\eps(B_{N+1}^\eps)\right)U_{N+1}.\cr}$$
Besides, by symbolic calculus, we obtain in ${\cal L}(L^2(\R^d))$
$$\displaylines{\quad
C^\eps(\delta)=U^\eps(\delta)^*\left[ \op_\eps\left(
Q_0+\Gamma^\eps_N
+\delta\eps^{N+1}(B_{N+1}^*QB_0+B_0^*QB_{N+1})+\delta\eps^{N+1}\tilde R_N\right)
 \right]U^\eps(\delta)\hfill\cr\hfill+O(\eps^{N+2})\qquad\cr\hfill
=U^\eps(\delta)^*\left[\op_\eps\left(Q_0+\Gamma_N^\eps+\eps^{N+1}\delta(D^*Q_0+Q_0D)+\delta\eps^{N+1}\tilde
R_N\right)\right]U^\eps(\delta) +O(\eps^{N+2})\cr}$$ where we have
used $B_0^*QB_0=Q_0$. Therefore, in view of \eqref{eq:U},
$$\frac{\d}{\d\,\delta}\,C^\eps(\delta)=\eps^{N+1}\,U^\eps(\delta)^*\op_\eps\left(\{f,Q_0\}+D^*Q_0+Q_0D+\tilde R_N\right)U^\eps(\delta)+O(\eps^{N+2})=O(\eps^{N+2})$$
in ${\cal L}(L^2(\R^d))$. Integrating between $\delta=0$ and $\delta =1$, we get
$$C^\eps(1)=C^\eps(0)+O(\eps^{N+2})\;\;{\rm in}\;\;{\cal L}(L^2(\R^d),$$
which gives the next step of the induction argument.

$ $

\ni It remains to prove that one can solve \aref{eq:homological}.
We write
$$\tilde R_N=p_N(0,z,0,\zeta)\,{\rm Id} + d_N(0,z,0,\zeta)\,J+s R_N^{(1)}+\sigma R_N^{(2)},\;\;J=\begin{pmatrix}1 & 0 \\ 0 & -1\end{pmatrix}.$$
In view of $\{f,Q_0\}=-\partial_s\, f{\rm Id} + \partial_\sigma f J$,  setting $$ f(s,z,\sigma,\zeta)=\tilde f(s,z,\sigma,\zeta)+
s\,p_N(0,z,0,\zeta)-\sigma \,
 d_N(0,z,0,\zeta)$$
 we are reduced to find $\tilde f$ and $D$ such that
\begin{equation}\label{eq:homological'}
\{\tilde f,Q_0\}+D^*Q_0+Q_0D+sR_N^{(1)} +\sigma R_N^{(2)}=0.
\end{equation}
 For this, we consider $Q_\delta
=Q_0+\delta\left(sR_N^{(1)} +\sigma R_N^{(2)}\right)$ for small $\delta$.
The strategy is to prove that   there exist a matrix $B_\delta$ and a canonical transform $\kappa_\delta$   such that
$$B_\delta^* \left(Q_\delta\circ \kappa_\delta\right) B_\delta=Q_0.$$
Differentiating this last relation with respect to $\delta$ and putting $\delta =0$, one solves the homological equation \aref{eq:homological'}.
Let us consider the symbol $Q_\delta$. We have
$$Q_\delta=\sigma\left({\rm Id}+\delta R_N^{(2)}\right)+s\left(J+\delta R_N^{(1)}\right).$$
The matrix ${\rm Id}+\delta R_N^{(2)}$ is symmetric and invertible
for $\delta$ small enough in a neighborhood of $0$. Thus, if
$C_1=\left(\sqrt{{\rm Id}+\delta R_N^{(2)}}\right)^{-1}$ then
$C_1=C_1^*$ and
$$C_1 Q_\delta C_1 =\sigma{\rm Id}+sC_1(J+\delta R_N^{(1)})C_1.$$
We observe that
$$ C_1(J+\delta R_N^{(1)})C_1=J+O(\delta),\;\;
{\rm tr}\left(C_1(J+\delta R_N^{(1)})C_1\right)=O(\delta),$$
therefore, we can write $ C_1(J+\delta R_N^{(1)})C_1=J+\delta R_N^{(3)}$ and
$$C_1\, Q_\delta\, C_1=\left(\sigma+s\delta \,{\rm tr}\left( R_N^{(3)}\right)\right) {\rm Id}+s\left(J+\delta R_N^{(4)}\right),
\;\;{\rm tr}(R_N^{(4)})=0.$$
One easily see that the matrix $C_1 Q_\delta C_1$ has a codimension $1$ crossing as long as $\delta$ is small enough so that the
eigenvalues of~$\delta R_N^{(4)}$ are smaller than $1$. Therefore, by the first step of our proof, there exist a matrix $B_1$ and a canonical transform $\kappa(\delta) $ such that
$$(B_1^* C_1 Q_\delta C_1 B_1) \circ \kappa(\delta)=Q_0.$$
Setting   $B(\delta)=(C_1 B_1)\circ \kappa(\delta)$, 
one  closes the proof of the normal form result. \qed
\end{proof}

\ni Thomas \textsc{Duyckaerts}, Universit\'e de Cergy Pontoise, 2 av. A. Chauvin, BP 222, 95302 Cergy-Pontoise Cedex, France.\\
{\tt Thomas.Duyckaerts@u-cergy.fr} 

\medskip

\ni Clotilde \textsc{Fermanian Kammerer}, Universit\'e Paris Est, 
UFR des Sciences et Technologie,
61, avenue du G\'en\'eral de Gaulle,
94010 Cr\'eteil Cedex, France.\\
{\tt Clotilde.Fermanian@univ-paris12.fr}

\medskip

\ni Thierry \textsc{Jecko}, Universit\'e de Cergy Pontoise,
2 av. A. Chauvin, BP 222, 95302 Cergy-Pontoise Cedex, France.\\
{\tt Thierry.Jecko@u-cergy.fr}

\begin{thebibliography}{99}


\bibitem{[BD]} P. Braam, H. Duistermaat: Normal forms of real symmetric systems with multiplicity. {\it Indag. Math.}; {\it N.S.} {\bf 4} (1993), no. 4, p.~407--421.

 \bibitem{Br} M. Brassart:  Limite semi-classique de transform\'ees de Wigner dans des milieux p\'eriodiques ou al\'eatoires. {\it
 Th\`ese de l'Universit\'e de Nice Sophia Antipolis}
(2002).



\bibitem{Bu} N. Burq: Semiclassical estimates for the resolvent
in non trapping geometry. {\it Int. Math. Res. Notices} {\bf 5}
(2002), p.~221--241.


\bibitem{Bu2} N. Burq: Smoothing effect for Schr\"odinger boundary value problem. {\it Duke Mathematical Journal}, {\bf 123}
(2004), no.~2, p.~221--241.

\bibitem{BuLe} N. Burq, G. Lebeau: Mesures de d\'efaut de compacit\'e, application au syst\`eme de Lam\'e. {\it  Ann. Sci. \'Ecole Norm. Sup. (4)}  {\bf 34} (2001), no.~6, p.~817--870.

\bibitem{CFMS} R. Carles, C. Fermanian Kammerer, N. Mauser, H.-P. Stimming: On the time evolution of Wigner measures for Schr\"odinger equations (to appear in  CPAA).

\bibitem{cjk} F. Castella, A. Knauf, Th. Jecko: Semiclassical
  resolvent estimates for Schr\"odinger operators with Coulomb singularities.
{\it Annales Henri Poincar\'e},  {\bf 9} (2008), p. 775-815. .

%
\bibitem{cj}F. Castella, Th. Jecko:  Besov estimates in the high-frequency 
Helmholtz equation, for a non-trapping and $C^2$ potential.
{\it J. Diff. Eq.},  {\bf 228} (2006), no.~2, 440-485. 
%

\bibitem{[CdV]} Y. Colin de Verdi\`ere:  The level crossing problem in semi-classical analysis. I. The symmetric case.
 Proceedings of the International Conference in Honor of Fr\'ed\'eric Pham (Nice, 2002).
 {\it  Ann. Inst. Fourier (Grenoble)} {\bf  53}  (2003),  no.~4, p.~1023--1054.

 \bibitem{[CdV2]} Y. Colin de Verdi\`{e}re: The level crossing problem in semi-classical analysis. II. The Hermitian case.
{\it  Ann. Inst. Fourier (Grenoble)}  {\bf 54}  (2004),  no.~5, p.~1423--1441.


\bibitem{CFKS} H. L. Cycon, R. Froese, W. Kirsch, B. Simon:
{\it Schr\"odinger operators with application to quantum mechanics
and global geometry}. Springer-Verlag (1987).


\bibitem{ds}M. Dimassi, J. Sj\"ostrand: {\it Spectral Asymptotics in the Semi-Classical Limit}. 
London Math. Soc. Lecture Note Series 268, Cambridge University Press (1999). 



\bibitem{Duy}  T.
Duyckaerts:
{ In\'egalit\'es de r\'esolvante pour l'op\'erateur de Schr\"odinger avec potentiel multipolaire critique.} (French.) [Resolvent estimates for the Schr\"odinger operator with critical multipolar potential]
{\it Bull. Soc. Math. France} {\bf 134} (2006), no.~2, p.~201--239.

\bibitem{F04} C. Fermanian  Kammerer: Semiclassical analysis of
generic codimension 3 crossings. {\it Int. Math. Res. Not.} {\bf 45}  (2004),
 p.~2391--2435.

 \bibitem{Fe03}
C.\ Fermanian  Kammerer:
 Wigner measures and molecular propagation through generic energy level crossings,
{\it Rev.\ Math.\  Phys.} {\bf 15} (2003), no.~10 , p.~1285--1317.



 \bibitem{FG} C. Fermanian Kammerer, P. G\'erard: Mesures
semi-classiques et croisements de modes. {\it Bull. Soc. math.
France} {\bf 130} (2002), no~1, p.~123--168.


 \bibitem{FG03} C. Fermanian Kammerer, P. G\'erard: A Landau-Zener formula for non-degenerated involutive codimension 3 crossings. {\it  Ann. Henri Poincar\'e}  {\bf 4}  (2003),  no~3,  p.~514--552.


\bibitem{FR} C. Fermanian Kammerer, V. Rousse: Resolvent estimates for a Schr\"odinger
operator with matrix-valued potential presenting eigenvalue crossings.
Application to Strichartz estimates, {\it Comm. in Part. Diff. Eq.} {\bf 33} (2008), no.~1, p.~19--44.


\bibitem{FH} R. Froese, I. Herbst: Exponential bounds and
absence of positive eigenvalues for $N$-body Schr\"odinger
operators. {\it Comm. Math. Phys. } {\bf 87}, 3 (1982/1983),
p.~429--447.

\bibitem{fln}S. Fujii\'e, C. Lasser, L. N\'ed\'elec :  Semiclassical resonances for a two-level Schr\"odinger 
operator with a conical intersection. Preprint.

\bibitem{Ge90} C. G\'erard: Semiclassical resolvent estimates for two and three-body Schr\"odinger operators.  {\it Comm. Partial Differential Equations}  {\bf 15}  (1990),  no.~8, p.~1161--1178.

\bibitem{GeMa}
C.\ G\'erard, A.~Martinez: Principe d'absorption limite pour des op\'erateurs de Schr\"odinger \`a longue port\'ee. (French) {\it [The limiting absorption principle for long-range Schr\"odinger operators]}  {\it C. R. Acad. Sci. Paris S\'er. I Math.}  {\bf 306} (1988),  no.~3, p.~121--123.

\bibitem{Ge93} P. G\'erard: Mesures semi-classiques et ondes de Bloch, {\it S\'eminaire sur les \'Equations aux D\'eriv\'ees Partielles, 1990--1991},
    Exp.\ No.\ XVI, 19
    (1991).

\bibitem{GeLe93} P. G\'erard, E. Leichtnam: Ergodic properties of eigenfunctions for the {D}irichlet problem, {\it Duke Math. J.} {\bf 71}, 2 (1993), p.~559--607.

\bibitem{GMMP}
P.\ G\'erard, P.\ Markowich, N.\ Mauser, and F.\ Poupaud:
Homogenization limits and Wigner transforms. {\it
Commun.\ Pure Appl.\ Math.} {\bf 50}, 4 (1997), p.~323--379.

\bibitem{GMMP2}
P.\ G\'erard, P.\ Markowich, N.\ Mauser, and F.\ Poupaud:
Erratum: ``Homogenization limits and Wigner transforms {\it
Commun.\ Pure Appl.\ Math.} {\bf 50}, 4 (1997), p.~323--379.'' {\it Commun.\ Pure Appl.\ Math.} {\bf 53} (2000), no.~2 p.~280--281.


%
\bibitem{ha}G. A. Hagedorn: {\em Molecular propagation through electron
energy level crossings.} Memoirs AMS, {\bf 536}, no.~111 (1994).
%

\bibitem{[Ho]}  L. H\"{o}rmander: {\em The analysis of linear Partial
Differential Operators III.} Springer-Verlag, 1985.


\bibitem{J1} T. Jecko: Estimations de la r\'esolvante pour
une mol\'ecule diatomique dans l'approximation de
Born-Oppenheimer. {\it Comm. Math. Phys.} {\bf 195} (1998) no.~3,
p.~585--612.

\bibitem{J2} T. Jecko: Semiclassical resolvent estimates for Schr\"odinger matrix 
operators with eigenvalues crossings. {\it Math. Nachr.}, {\bf 257} (2003), no.~1, p. 36-54.

\bibitem{J5} T. Jecko: From classical to semiclassical non-trapping
behaviour. {\it C. R. Acad. Sci. Paris}, Ser.~I {\bf 338} (2004),
p.~545--548.

\bibitem{J6} T. Jecko: Non-trapping condition for semiclassical
Schr\"odinger operators with matrix-valued potentials. {\it Math.
Phys. Electronic Journal}  {\bf 11} (2005), no.~2.\\
{\bf Erratum}: {\it Math. Phys. Electronic Journal}, No. 3, vol. {\bf 13}, 2007.  



\bibitem{Kar} U. Karlsson: Semi-Classical approximations of Quantum Mechanical Problems. Doctoral Dissertation (2002), Royal Institute of Technology, Stockolm.

\bibitem{Kri} J-L. Krivine: {\it Introduction to axiomatic set theory.} Translated from the French by David Miller D. Reidel Publishing Co., Dordrecht; Humanities Press, New York 1971, 98 pp.

\bibitem{Kri2} J-L. Krivine: {\it Th\'eorie des ensembles.} Cassini, Paris 2007. 

\bibitem{Le} G. Lebeau: Equation des ondes amorties, in {\it Algebraic and
geometric methods in mathematical physics }(Kaciveli, 1993), {\it Math.
Phys. Stud.}, vol. 19, Kluwer Acad. Publ., Dordrecht, 1996, p.~
73--109.

\bibitem{LP} P.-L. Lions, T. Paul: Sur les mesures de {W}igner, {\it Revista Matem\'atica Iberoamericana}
     {\bf 9} (1993), no.~3, p.~553--618.

%
\bibitem{ma} A. Martinez: {\em An introduction to semiclassical and microlocal analysis.} 
Universitext Springer, 2002. 
%

\bibitem{MS} R. B. Melrose, J. Sj\"ostrand:  Singularities of boundary value problems. I. {\it Comm. Pure Appl. Math.} {\bf 31}  (1978), no.~5, p.~593--617.


\bibitem{Mi} L. Miller:
    Propagation d'onde semi-classiques \`{a}
travers une interface et mesures 2-microlocales. {\it Th\`{e}se de
l'Ecole Polytechnique} (1996).

\bibitem{Mo} E. Mourre: Absence of singular continuous spectrum
for certain self-adjoint operators. {\it Comm. Math. Phys.} {\bf
78} (1981), p.~391--408.

%
\bibitem{ne}L. N\'ed\'elec: Resonances for matrix Schr\"odinger 
operators. {\it Duke Math. J.} {\bf 106} (2001), no.~2, p. 209-236. 
%

\bibitem{RS2} M. Reed, B. Simon: {\em Method of Modern Mathematical
Physics, Tome~II~: Fourier Analysis, Self-adjointness.} Academic
Press 1979.

\bibitem{[RS4]} M. Reed, B. Simon: {\em Method of Modern Mathematical
Physics, Tome~IV~: Analysis of operators.} Academic
Press 1979.



\bibitem{[Ro]} D. Robert:  {\em Autour de l'approximation semi--classique.}
{ Birkha\"user}, 1983.

%
\bibitem{rt}D. Robert, H. Tamura:  Semiclassical estimates for
resolvents and asymptotics for total cross-section. {\it Ann. IHP} {\bf 46} (1987),
p.~415--442.
%
\bibitem{VaZw00} A.\ Vasy, M.\ Zworski:
Semiclassical estimates in asymptotically Euclidean scattering.{\it
Comm. Math. Phys.} {\bf 212} (2000), no.~1, p.~205--217.

\bibitem{Vaught} R. L. Vaught: {\em Set Theory. An introduction.} Birkh\"auser Boston, 1995.

\bibitem{w87} X.P. Wang:
Time-decay of scattering solutions and classical trajectories. {\it
Ann. Inst. H. Poincar\'e Phys. Th\'eor.} {\bf 47} (1987), no.~1, p.~25--37.



%
\bibitem{w}X.P. Wang: Semiclassical resolvent estimates for
$N$-body Schr\"odinger operators.  {\it J. Funct. Anal.}
{\bf  97} (1991), p.~466--483.
%


\end{thebibliography}
\end{document}